\documentclass[aps,pra,twocolumn,groupedaddress]{revtex4}

\usepackage[latin1]{inputenc}
\usepackage[T1]{fontenc}
\usepackage[english]{babel}
\usepackage[dvips]{graphicx}
\usepackage{enumerate,amsthm,amsmath,amssymb,color,graphicx,bbm,float,framed}
\usepackage{natbib}

\makeatletter
\newtheorem*{rep@theorem}{\rep@title}
\newcommand{\newreptheorem}[2]{%
\newenvironment{rep#1}[1]{%
 \def\rep@title{#2 \ref{##1}}%
 \begin{rep@theorem}}%
 {\end{rep@theorem}}}
\makeatother

\makeatletter

\renewenvironment{widetext@grid}{%
  \par\ignorespaces
  \setbox\widetext@top\vbox{%
   \vskip15\p@
   \hb@xt@\hsize{%
    \leaders\hrule\hfil
    \vrule\@height6\p@
   }%
   \vskip6\p@
  }%
  \setbox\widetext@bot\hb@xt@\hsize{%
    \vrule\@depth6\p@
    \leaders\hrule\hfil
  }%
  \onecolumngrid
  \let\set@footnotewidth\set@footnotewidth@ii
}{%
  \par
  \twocolumngrid\global\@ignoretrue
  \@endpetrue
}%

\makeatother

\newtheorem{theo}{Theorem} 
\newtheorem{lemma}[theo]{Lemma}

\newreptheorem{lemma}{Lemma}

\newtheorem{cor}[theo]{Corollary}

\usepackage{epstopdf,hyperref}

\newcommand{\be}{\begin{equation}}
\newcommand{\ee}{\end{equation}}

\usepackage{mathtools}
\def\multiset#1#2{\ensuremath{\left(\kern-.3em\left(\genfrac{}{}{0pt}{}{#1}{#2}\right)\kern-.3em\right)}}

\makeatletter

\renewenvironment{widetext@grid}{%
  \par\ignorespaces
  \setbox\widetext@top\vbox{%
   \vskip15\p@
   \hb@xt@\hsize{%
    \leaders\hrule\hfil
    \vrule\@height6\p@
   }%
   \vskip6\p@
  }%
  \setbox\widetext@bot\hb@xt@\hsize{%
    \vrule\@depth6\p@
    \leaders\hrule\hfil
  }%
  \onecolumngrid
  \let\set@footnotewidth\set@footnotewidth@ii
}{%
  \par
  \twocolumngrid\global\@ignoretrue
  \@endpetrue
}%

\makeatother

\begin{document}

\title{Composable security proof for continuous-variable quantum key distribution with coherent states}

\author{Anthony Leverrier}
\address{Inria, EPI SECRET, B.P. 105, 78153 Le Chesnay Cedex, France}
\email{anthony.leverrier@inria.fr}

\date{\today}

\begin{abstract}
We give the first \emph{composable} security proof for continuous-variable quantum key distribution with coherent states against collective attacks. Crucially, in the limit of large blocks the secret key rate converges to the usual value computed from the Holevo bound. 
Combining our proof with either the de Finetti theorem or the Postselection technique then shows the security of the protocol against general attacks, thereby confirming the long-standing conjecture that Gaussian attacks are optimal asymptotically in the composable security framework.

We expect that our parameter estimation procedure, which does not rely on any assumption about the quantum state being measured, will find applications elsewhere, for instance for the reliable quantification of continuous-variable entanglement in finite-size settings.
\end{abstract}

\maketitle

Quantum key distribution (QKD) is a cryptographic primitive that allows two distant parties, Alice and Bob, who have access to an insecure quantum channel and an authenticated classical channel, to distill a secret key. QKD has spurred a lot of interest in the past decades because it is arguably the first application of the field of quantum information to reach commercial maturity \cite{SBC09}.
Despite a lot of effort invested in the theoretical analysis of QKD protocols, \emph{composable security} \cite{can01,MR09} has only been established for a handful of protocols, for instance BB84 \cite{BB84}. This major achievement is the latest step in a series of more and more refined security proofs and improved bounds for the secret key rates. More precisely, composable security proofs have successively used an exponential version of the de Finetti theorem \cite{ren07}, the Postselection technique \cite{CKR09} and an entropic uncertainty principle \cite{TLG12}. 

The situation for continuous-variable (CV) protocols is much less advanced \cite{WPG12}. These protocols \cite{GG02, WLB04}, which do not require single-photon detectors, are particularly appealing in terms of implementation \cite{JKL13} but their security is still far from being completely understood. 
Recently, a composable security proof for a CV protocol was obtained \cite{FBB12, fur14,GHD14} from an entropic uncertainty principle \cite{BCF13} but the protocol requires the generation of squeezed states and is only moderately tolerant to losses. 
Other approaches to establish the security of a protocol typically consist of two independent steps: first a composable security proof valid against collective attacks, a restricted type of attacks where the quantum state shared by Alice and Bob protocol displays a tensor product structure, followed by an additional argument to obtain security against general attacks. 
These two steps have been partially completed in the case of CV protocols: a reduction from general to collective attacks is obtained via two possible techniques, namely a de Finetti theorem \cite{RC09}, and the Postselection technique \cite{LGR13}, the latter technique being more efficient but at the price of adding an unpractical symmetrization step to the protocol (The Postselection technique should not be mistaken with the postselection of data in certain protocols \cite{SRL02}.). Unfortunately, security against collective attacks has only been proved (via a Gaussian optimality argument \cite{WGC06}) in the asymptotic limit, which does not say anything about composable security \cite{GC06, NGA06,PBL08}. Note also that finite-size effects for CV QKD were partly explored in Ref.~\cite{LGG10}, but under a Gaussian attack assumption.

In this paper, we give the first composable security proof valid against collective attacks for CV QKD with coherent states (We note that coherent states can also be used in BB84 implementations, for instance in decoy-state protocols, and that composable security has been proved in Ref.~\cite{HN14}.)  and either direct or reverse reconciliation \cite{GVW03}. The Postselection technique then implies composable security against general attacks. Remarkably, the secret key rate is asymptotically equal to the one assuming a Gaussian attack, which is not the case for the proof based on the uncertainty principle. This is crucial for the distribution of keys over long distances \cite{JKL13}.

To prove this result, we develop a number of techniques including a tool for reliable tomography of the covariance matrix without making any assumption about the quantum state.
By performing the Parameter Estimation (PE) step after Error Correction (EC), we improve the estimation and are able to use almost all the raw data to distill the secret key. A similar strategy was also considered for BB84 in Ref.~\cite{FMC10}. Our only assumptions are that Alice and Bob have access to a classical authenticated channel and that their equipment is trusted: they can prepare coherent states and detect light with heterodyne detection. Our framework can easily incorporate imperfections either in the preparation or in the detection, as long as they are properly modeled. To keep the notations simple, we will however assume that the equipment of the legitimate parties is perfect.

{\bf Composable security}.--- An Entanglement-Based (EB) QKD protocol $\mathcal{E}$ is a Completely-Positive Trace-Preserving (CPTP) map $\mathcal{E}\colon \mathcal{H}_A \otimes \mathcal{H}_B  \rightarrow  \mathcal{S}_A \otimes \mathcal{S}_B \otimes \mathcal{C}$ that takes an arbitrary input state $\rho_{AB}$ shared by Alice and Bob and outputs for each party a classical string $S_A$ or $S_B$, and some public transcript $C$. For a CV protocol, both $\mathcal{H}_A$ and $\mathcal{H}_B$ correspond to infinite-dimensional Fock spaces, while $\mathcal{S}_A, \mathcal{S}_B$ and $\mathcal{C}$ describe classical registers.

A QKD protocol should be \emph{secure}, meaning that the output keys should be identical and secret \cite{TLG12}. It should also be \emph{robust}, i.e. output nontrivial keys if there is no active attack on the quantum channel. These are actually properties of the output state of the protocol, more precisely of $\rho_{S_A S_B E}$, which should hold for any input state. In this paper, we denote by $\rho_{\mathcal{H}}$ the marginal of the state $\rho$ restricted to subspace $\mathcal{H}$.
The subscript $E$ refers to the quantum register $\mathcal{H}_E$ of the adversary, and the final state is obtained by applying the map $\mathcal{E} \otimes \mathrm{id}_{\mathcal{H}_E}$ to an arbitrary purification $\Psi_{ABE}$ of $\rho_{AB}$.
A QKD protocol is called \emph{correct} if $S_A = S_B$ for any strategy of the adversary, that is, any initial state of the protocol $\Psi_{ABE}$. A protocol is $\epsilon_{\mathrm{cor}}$-correct if $\mathrm{Pr}[S_A \ne S_B] \leq \epsilon_{\mathrm{cor}}$.
Denote by $\mathcal{H}_{E'} = \mathcal{H}_E \otimes  \mathcal{C}$ the space accessible to the adversary (her quantum system $E$ and the public transcript $\mathcal{C}$). 
A key is called $\delta$-\emph{secret} if it is $\delta$-close to a uniformly distributed key that is uncorrelated with the eavesdropper:
\begin{align}
\frac{1}{2} \left\| \rho_{S_A^l  E'} - \omega_{l} \otimes \rho_{E'} \right\|_1 \leq \delta,
\end{align}
where $\rho_{S_AE'}^l$ is the state conditioned on the key length $l$ and  $\omega_{l}$ is the fully mixed state on classical strings of length $l$.
If the protocol aborts, it outputs a dummy key of size 0, which is automatically secret. A QKD protocol is called $\epsilon_{\mathrm{sec}}$-secret if for any attack strategy, it outputs $\delta$-secret keys with $(1-p_{\mathrm{abort}}) \delta \leq \epsilon_{\mathrm{sec}}$, where $p_{\mathrm{abort}}$ is the abort probability.
A QKD protocol is $\epsilon$-\emph{secure} if it is $\epsilon_{\mathrm{sec}}$-secret and $\epsilon_{\mathrm{cor}}$-correct with $\epsilon_{\mathrm{sec}} + \epsilon_{\mathrm{cor}} \leq \epsilon$.
Since a protocol that would always abort is perfectly secure according to this definition, it is important to take into account its \emph{robustness} $\epsilon_{\mathrm{rob}}$, which is the probability that the protocol aborts if the eavesdropper is inactive. In the case of a CV QKD protocol, this corresponds to a thermal bosonic channel, which is a good model for the transmission of light in an optical fiber.

{\bf Description of the CV QKD protocol $\mathcal{E}_0$}.--- We focus here on the EB version of the protocol, but the security of its Prepare and Measure (PM) version where Alice sends coherent states and Bob uses heterodyne detection follows immediately. Moreover, we present the reverse reconciliation version, which is the most useful in practice. The direct reconciliation version is easily obtained by interchanging the roles of Alice and Bob in the classical post-processing part of the protocol. 
Recall that in order to obtain security against general attacks, one would need to add another step to the protocol, involving an energy test as well as a potential symmetrization procedure.

The protocol $\mathcal{E}_0$ is sketched in Fig.~\ref{protocol-E0} (and detailed in the appendix) and depends on a number of parameters: most notably, the number $2n$ of coherent states sent by Alice, the length $l$ of the final key if the protocol did not abort, the discretization parameter $d$, the size of Bob's communication to Alice, $\mathrm{leak}_{\mathrm{EC}}$, during the error correction procedure, the maximum failure probabilities $\epsilon_{\mathrm{cor}}$ and $\epsilon_{\mathrm{PE}}$ for the EC and PE steps, respectively, some bounds on covariance matrix elements, $\Sigma_a^{\max}, \Sigma_b^{\max}, \Sigma_c^{\min}$ for the PE test to pass and a robustness parameter $\epsilon_{\mathrm{rob}}$.

Our main result quantifies the security of the protocol $\mathcal{E}_0$ in the composable security framework.
\begin{theo} 
\label{key-rate-theorem}
The protocol $\mathcal{E}_0$ is $\epsilon$-secure against collective attacks if $\epsilon =  \sqrt{\epsilon_{\mathrm{PE}} + \epsilon_{\mathrm{cor}} + \epsilon_{\mathrm{ent}}} + 2\epsilon_{\mathrm{sm}} + \bar{\epsilon} $ and if the key length $l$ is chosen such that
\begin{align}
l \leq & 2n \left[ 2\hat{H}_{\mathrm{MLE}}(U) - f(\Sigma_a^{\max}, \Sigma_b^{\max}, \Sigma_c^{\min}) \right]- \mathrm{leak}_{\mathrm{EC}}  \nonumber\\
& -  \Delta_{\mathrm{AEP}} - \Delta_{\mathrm{ent} } - 2\log\frac{1}{2\bar{\epsilon}},
\label{key-rate}
\end{align}
where $\hat{H}_{\mathrm{MLE}}(U)$ is the empirical entropy of $U$, $\Delta_{\mathrm{AEP}} := \sqrt{2n} \left[(d+1)^2 + 4(d+1) \log_2 \frac{2}{\epsilon_{\mathrm{sm}}^2} + 2\log_2 \frac{2}{\epsilon^2\epsilon_{\mathrm{sm}}}\right] +4 \epsilon_{\mathrm{sm}} d/\epsilon$, $\Delta_{\mathrm{ent}} :=   \log_2 \frac{1}{\epsilon}+ \sqrt{8n \log_2^2 (4 n) \log(2/\epsilon_{\mathrm{sm}})}$ and $f$ is the Holevo information between Eve and Bob's  measurement result for a Gaussian state with covariance matrix parametrized by $\Sigma_a^{\max}, \Sigma_b^{\max}, \Sigma_c^{\min}$.
\end{theo}

\begin{figure}[htbp]
\begin{framed}
  \centering
\begin{enumerate}
\item \textbf{State Preparation:} Alice prepares $2n$ two-mode squeezed vacuum states, keeps the first half of each state and transmits the second half to Bob through an insecure quantum channel. Alice and Bob then share a global quantum state $\rho_{AB}^{\otimes (2n)}$.
\item \textbf{Measurement:} Alice and Bob measure their respective modes with heterodyne detection and obtain two strings $X, Y \in \mathbbm{R}^{4n}$. Bob discretizes his $4n$-vector $Y$ to obtain the $m$-bit string $U$, where $m = 4dn$, i.e. each symbol is encoded with $d$ bits of precision.
\item \textbf{Error Correction:} Bob sends some side information of size $\mathrm{leak}_{\mathrm{EC}}$ to Alice (syndrome of $U$ for a linear error correcting code $C$ agreed on in advance) and Alice outputs a guess $\hat{U}$ for the string of Bob. Bob computes a hash of $U$ of length $\lceil \log_2(1/\epsilon_{\mathrm{cor}}) \rceil$ and sends it to Alice who compares it with her own hash. If both hashes differ, the protocol aborts. 
\item \textbf{Parameter Estimation:} Bob sends $n_{\mathrm{PE}}=O(\log(1/\epsilon_{\mathrm{PE}}))$ bits of information to Alice that allow her to compute $\|X\|^2, \|Y\|^2$ and $\langle X,Y\rangle$, as well as $\gamma_a, \gamma_b$ and $\gamma_c$ defined in Eq.~\ref{gamma_a_def}, \ref{gamma_b_def}, \ref{gamma_c_def}. The PE test passes if $[\gamma_a \leq \Sigma_a^{\mathrm{max}}] \wedge[\gamma_b \leq \Sigma_b^{\mathrm{max}}] \wedge [\gamma_c\geq \Sigma_c^{\mathrm{min}}]$; otherwise the protocol aborts.  
\item \textbf{Privacy Amplification:} Alice and Bob apply a random universal$_2$ hash function to their respective strings, obtaining two strings $S_A$ and $S_B$ of size $l$. 
\end{enumerate} 
\end{framed}
\caption{Protocol $\mathcal{E}_0$, with reverse reconciliation and parameters $n, l, \mathrm{leak}_{\mathrm{EC}}, \epsilon_{\mathrm{cor}}, n_{\mathrm{PE}}, \epsilon_{\mathrm{PE}}, \Sigma_a^{\max}, \Sigma_b^{\max}, \Sigma_c^{\min},  d$ }
\label{protocol-E0}
\end{figure}

This secret key size should be compared to the asymptotic secret key rate assuming collective, Gaussian attacks. This can be done by assuming a passive quantum channel corresponding to a Gaussian channel with transmittance $T$ and excess noise $\xi$. One needs to factor in the robustness of the protocol, that is the probability that the PE test will not pass in the case of a passive channel. We plot the secret key rate as a function of $n$ for $\epsilon = 10^{-20}$ on Fig.~\ref{fig:key-rate-small}. The asymptotic key rate is typically reached for $n$ between $10^8$ and $10^{11}$ for distances up to 50 km. 
\begin{figure}
\centering
  \includegraphics[width=1.\linewidth]{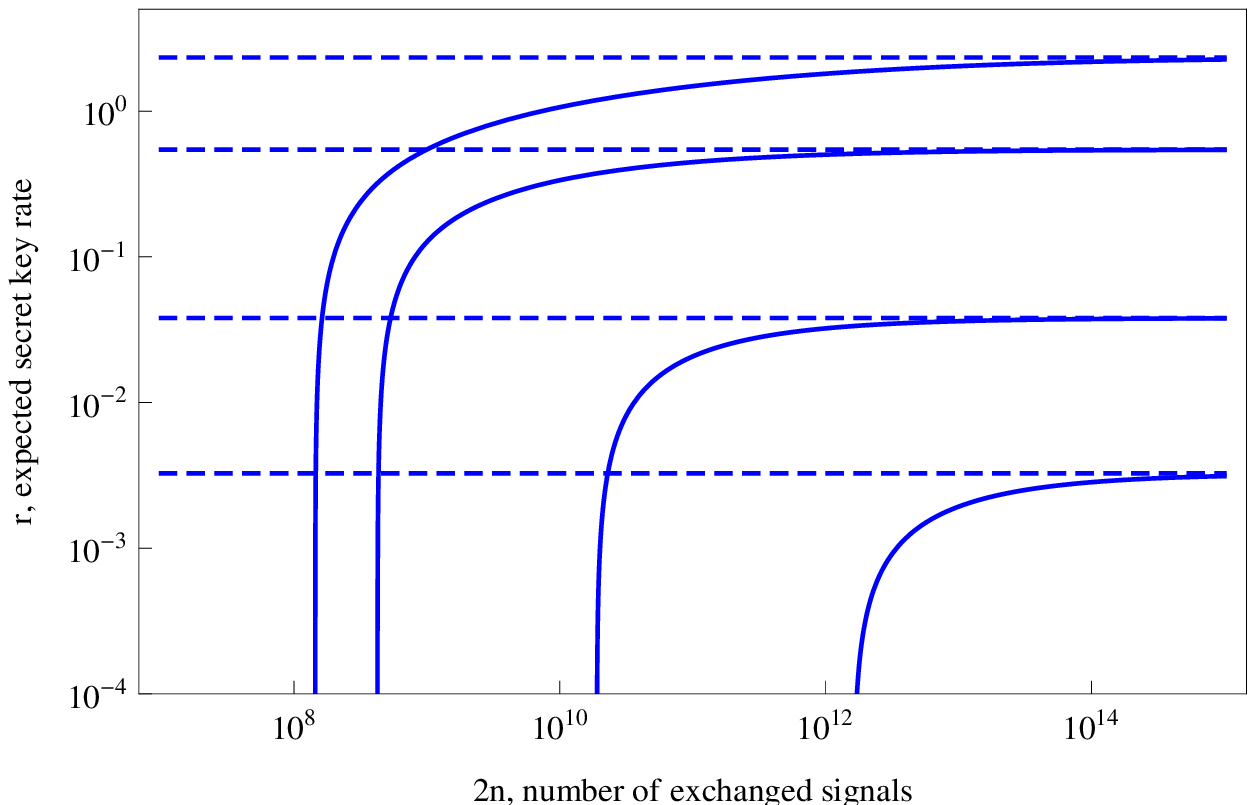}
  \label{fig:scheme1}
\caption{Expected secret key rate $r = (1-\epsilon_{\mathrm{rob}})l/2n$ secure against collective attacks, as a function of $2n$, the number of exchanged signals. From top to bottom, the transmittance of the quantum channel corresponds to distances of 1 km, 10 km, 50 km and 100 km for assumed losses of $0.2$ dB per km. For each distance, the expected secret key rate reaches the asymptotic value for large enough $n$. The modulation variance is optimized, the reconciliation efficiency is set to $\beta = 0.95$, the discretization parameter to $d=5$ (the value of $d$ should be optimized depending on the error correcting codes used in the reconciliation, see e.g. \cite{JEK14}), the excess noise to $\xi  =0.01$, the robustness parameter to $\epsilon_{\mathrm{rob}} \leq 10^{-2}$ and the security parameter to $\epsilon = 10^{-20}$. Dashed lines correspond to the respective asymptotic expected secret key rates. Refer to the Supplementary Material for a detailed derivation of the value of the expected secret key rate.}
\label{fig:key-rate-small}
\end{figure}

{\bf Parameter Estimation}.--- 
We defer the detailed description of the protocol $\mathcal{E}_0$ and its full security proof to the appendix and focus more specifically on the PE step here. A novelty of the protocol is that PE is performed \emph{after} EC. This can be done quite efficiently since a rough estimate of the signal-to-noise ratio (SNR) of the data is in general sufficient to choose an appropriate error correcting code and proceed with the reconciliation. At the end of the EC step, Alice therefore knows the strings $X$ and $U$ and it is not hard to show that if Bob sends her a few additional bits, she can learn the values of $\|X\|^2, \|Y\|^2$ and $\langle X, Y\rangle$ arbitrarily well.

The goal of the PE step is to obtain a confidence region for the covariance matrix of the state $\rho_{AB}^{\otimes (2n)}$. Here, one needs to be careful because by the time PE is performed, the state has already been measured and it does not make real sense to talk about its covariance matrix anymore. We will follow the paradigm for tomography introduced in Ref.~\cite{CR12} and define a quantum tomography process as a CPTP map that takes an input state $\rho^{n+k} \in \mathcal{H}^{\otimes (n+k)}$, symmetrizes it and outputs a state $\rho^{n} \in \mathcal{H}^{\otimes n}$ as well as a confidence region $R$ of $\mathcal{P}_=(\mathcal{H}^{\otimes n})$, the set of normalized density operators on $\mathcal{H}^{\otimes n}$. The superscript $n$ for $\rho^n$ should not be interpreted as saying that the state has an i.i.d.~structure, nor that $\rho^n$ corresponds to a marginal state of $\rho^{n+k}$; it is merely a remainder of the size of Hilbert space it lives in. In words, it consists in measuring a subsystem of the initial state and making a prediction for the remaining state.  The quality of the quantum tomography is assessed by two parameters: the probability  that the prediction is false and the size of the region. A larger region means a smaller error probability, but also a more pessimistic secret key rate. 

An important issue concerning the tomography of a CV system is that the covariance matrix is \textit{a priori} unbounded. 
Consider for instance the state $\sigma^{\otimes (n+k)}$ with $\sigma = (1-\epsilon) |0\rangle \langle 0| + \epsilon |N\rangle \langle N|$. The covariance matrix of $\sigma$ is $\mathrm{diag}(1+ N\epsilon/2, 1+N\epsilon/2)$ but any tomographic procedure that only examines $k \ll 1/\epsilon$ modes will conclude that the covariance matrix is close to that of the vacuum, which is clearly incorrect if $N\epsilon \gg1$.
The solution to this problem consists in first appropriately symmetrizing the state $\rho^{n+k}$ before measuring $k$ subsystems and inferring properties for the remaining $n$ modes.

Ideally, the tomography of the input state $\rho_{AB}^{2n}$ of the QKD protocol $\mathcal{E}_0$ should consist of the following steps, which involve additional parties $A_1$ and $A_2$ on Alice's side and $B_1$ and $B_2$ on Bob's side:
\begin{enumerate}
\item \emph{State symmetrization:} Alice's $2n$ modes are processed with a random network of beamsplitters and phase-shifts, and Bob's modes with the conjugate network, giving a new state $\tilde{\rho}^{2n}$.
\item \emph{Distribution to additional players:} Alice and Bob distribute $\tilde{\rho}_1^n$ corresponding to the first $n$ modes of $\tilde{\rho}^{2n}$ to $A_1$ and $B_1$. Similarly, they give $\tilde{\rho}_2^n$ to $A_2$ and $B_2$.
\item \emph{Measurement:} $A_1$ and $B_1$ measure $\tilde{\rho}_1^n$ with heterodyne detection and obtain two vectors $X_1, Y_1 \in \mathbbm{R}^{2n}$. Similarly, $A_2$ and $B_2$ obtain $X_2, Y_2 \in \mathbbm{R}^{2n}$.
\item \emph{Parameter Estimation:}
$B_1$ sends some information to $A_1$ so that she can learn the values of $\|X_1\|^2, \|Y_1\|^2$ and $\langle X_1, Y_1\rangle$ and then compute a confidence region for the (averaged) covariance matrix of $\tilde{\rho}_2^n$. Similarly, $A_2$ computes a confidence region for that of $\tilde{\rho}_1^n$. 
\end{enumerate}
By averaged covariance matrix, we mean the three real values $\Sigma_a, \Sigma_b, \Sigma_c$ defined by:
$\Sigma_{a/b}  := \frac{1}{2n} \sum_{i=1}^n \left( \langle q_{A/B_i}^2 \rangle + \langle p_{A/B,i}^2\rangle\right)$ and $\Sigma_c  := \frac{1}{2n}  \sum_{i=1}^n  \left(\langle q_{A,i} q_{B,i} \rangle - \langle p_{A,i} p_{B,i}\rangle\right)$
where $q_{A,i}$ is the quadrature operator $\frac{1}{\sqrt{2}} (\hat{a}_i + \hat{a}^\dagger_i)$ for the $i^{\mathrm{th}}$ mode of Alice for instance. 

An interesting feature of this PE procedure is that $A_1$ and $A_2$ can respectively estimate the covariance matrices of $\tilde{\rho}_2^n$ and $\tilde{\rho}_1^n$, meaning that a secret key can be distilled from both halves of the state. In other words, no raw key is wasted because of parameter estimation.
While it is clear that this scheme is rather impractical, one can show that it can nevertheless be efficiently simulated by Alice, without any need for symmetrization or for additional parties. 

In fact, if Alice learns the values of $\|X\|^2, \|Y\|^2$ and $\langle X,Y\rangle$, she can compute $\gamma_a, \gamma_b, \gamma_c$ as follows
\begin{align}
\gamma_a &:= \frac{1}{2n} \left[ 1 + 2\sqrt{\frac{\log(36/\epsilon_{\mathrm{PE}})}{n}}\right] \|X\|^2-1,\label{gamma_a_def}\\
\gamma_b &:= \frac{1}{2n} \left[ 1 + 2\sqrt{\frac{\log(36/\epsilon_{\mathrm{PE}})}{n}}\right] \|Y\|^2-1\label{gamma_b_def}\\
\gamma_c &:= \frac{1}{2n} \langle X, Y\rangle - 5 \sqrt{\frac{\log (8/\epsilon_{\mathrm{PE}})}{n^3}}(\|X\|^2 + \|Y\|^2).\label{gamma_c_def}
\end{align}
We are now in a position to define the Parameter Estimation Test and bound its failure probability (proven in the appendix).
\begin{theo}
\label{PE-test}
The probability that the Parameter Estimation Test passes, that is, $[\gamma_a \leq \Sigma_a^{\max} ] \wedge [\gamma_b \leq \Sigma_b^{\max} ] \wedge [\gamma_c \geq \Sigma_c^{\min} ] $ and that Eve's information $\chi(U;E)$ computed for the Gaussian state with covariance matrix characterized by $\Sigma_a^{\max}, \Sigma_b^{\max}$ and $\Sigma_c^{\min}$ is underestimated is upper-bounded by $\epsilon_{\mathrm{PE}}$. 
\end{theo}

Here the Holevo information $\chi(U;E)$ is upper bounded by $f(\Sigma_a^{\max}, \Sigma_b^{\max}, \Sigma_c^{\min})  := g[(\nu_1-1)/2]+g[(\nu_2-1)/2]-g[(\nu_3-1)/2]$ where $\nu_1$ and $\nu_2$ are the symplectic eigenvalues of the covariance matrix $\left[
\begin{smallmatrix}
\Sigma_a^{\mathrm{max}} \mathbbm{1}_2 & \Sigma_c^{\mathrm{min}} \sigma_z \\
\Sigma_c^{\mathrm{min}}  \sigma_z & \Sigma_b^{\mathrm{max}}\mathbbm{1}_2 \\
\end{smallmatrix}
\right]$, $\nu_3 = \Sigma_a^{\mathrm{max}}- \left(\Sigma_c^{\mathrm{min}} \right)^2/(1+ \Sigma_b^{\mathrm{max}})$, $\sigma_z = \mathrm{diag}(1,-1)$ and  $g(x) := (x+1) \log_2(x+1) -x \log_2 (x)$.

Once we are able to analyze the PE test, the rest of the security proof follows in a rather straightforward fashion: see the appendix for all the details. 
It should be noted that the assumption of collective attacks was not used in the PE step: this is because the symmetrization breaks the tensor product of the state. However, we crucially rely on the collective attack assumption when exploiting the Asymptotic Equipartition Property of the smooth min-entropy, which is the quantity of interest to analyze the success of the privacy amplification step.

{\bf A security proof against general attacks}.--- 
So far, we have restricted the analysis to collective attacks. For CV QKD, there are two known techniques to obtain a full security proof from one holding against collective attacks: an exponential version of de Finetti theorem \cite{RC09} and the Postselection technique \cite{LGR13}. The former technique directly applies here and can be used to upgrade the protocol $\mathcal{E}_0$ to a slightly more complicated one (including an energy test and a random permutation) that is provably $\tilde{\epsilon}$-secure against general attacks, but with $\tilde{\epsilon} \gg \epsilon$, provided the key length is adequately shortened. However, while this provides composable security CV QKD with coherent states against general attacks, it does not give very good finite-key estimates. 
The Postselection technique is better, but still falls short on providing useful finite-size key estimates. 
Indeed, in order to apply it, one needs to add an energy test which depends on a small parameter $\epsilon_{\mathrm{test}}$, and if the protocol $\mathcal{E}_0$ was $\epsilon$-secure against collective attacks, the new protocol is $\tilde{\epsilon}$-secure against general attacks where $\tilde{\epsilon} = \epsilon 2^{O(\log^4(n/\epsilon_{\mathrm{test}}))} + 2 \epsilon_{\mathrm{test}}$, which is prohibitive in practice. 
Moreover, in the case of reverse reconciliation, it seems that the current Postselection technique requires an additional symmetrization step for the classical data, which has complexity $\Theta(n^2)$. Whether or not this symmetrization can be simulated, as was the case in the PE step, is left as an interesting open question. 

{\bf Conclusion}.--- 
We have provided a composable security proof of a CV QKD protocol using coherent states valid against collective attacks. This was the missing step to establish the security of such protocols against general attacks in the composable security framework. The bounds we obtained are compatible with state-of-the-art experiments. For protocols with direct reconciliation, this directly gives a composable security proof against general attacks. For reverse reconciliation, which is required to achieve long distances, an additional symmetrization step provides the same level of security. 
Further work will be needed to improve the current reductions from general to collective attacks, which should be possible since the current techniques do not exploit all the symmetries of the protocols.

We expect our parameter estimation procedure to find applications in the field of continuous-variable entanglement. Indeed, most criteria for detecting CV entanglement are based on the covariance matrix \cite{DGC00} and to our knowledge, our procedure gives the first robust estimation of the covariance matrix of an unknown quantum state without relying on any assumption such as the Gaussian nature of the state.

{\bf Acknowledgements}.--- I thank Fabian Furrer and Philipe Grangier who provided very useful comments on a preliminary version of this manuscript.  I'm especially grateful to Marco Tomamichel for enlightening discussions about smooth entropies and the asymptotic equipartition property.

\newpage
\appendix 
\begin{widetext}

In this appendix, we first recall the basic definitions related to the composable security of QKD (Section \ref{presentation}). Then, in Section \ref{protocol}, we give a complete description of the QKD protocol $\mathcal{E}_0$ that is secure against collective attacks. In Section \ref{expected}, we explain in details how to compute the expected secret key rate. In Section \ref{tools}, we introduce the tools needed for the security proof. In Section \ref{section-PE}, we explain the ideas behind the Parameter Estimation procedure and establish an upper bound for the failure probability of the PE test.  We present the security proof in Section \ref{sec-proof}. Finally, in Section \ref{general-attacks}, we explain how to upgrade $\mathcal{E}_0$ to a protocol provably secure against general attacks.

\section{Quantum Key Distribution and Composable Security}
\label{presentation}

An Entanglement-Based (EB) QKD protocol $\mathcal{E}$ is a Completely-Positive Trace-Preserving (CPTP) map:
\begin{align}
\begin{array}{cccc}
\mathcal{E}\colon &\mathcal{H}_A \otimes \mathcal{H}_B  &\rightarrow & \mathcal{S}_A \otimes \mathcal{S}_B \otimes \mathcal{C}\\
& \rho_{AB} & \mapsto & \rho_{S_A, S_B, C}.
 \end{array}
\end{align}
Let us explain the notations. The Hilbert spaces $ \mathcal{H}_A$ and $\mathcal{H}_B $ refer to the state spaces of Alice and Bob, respectively. Typically, each of these spaces corresponds to a large number of copies of a small Hilbert space. For instance, in the case of the BB84 protocol, one has $\mathcal{H}_A \cong \mathcal{H}_B \cong (\mathbbm{C}^{2})^{\otimes n}$, where $n$ is a large number counting how many 2-qubit states are shared by Alice and Bob.
 For a Hilbert space $\mathcal{H}$, we denote by $\mathcal{P}(\mathcal{H})$ the set of positive semi-definite operators on $\mathcal{H}$ and by $\mathcal{P}_=(\mathcal{H})$ the set of operators in $\mathcal{P}(\mathcal{H})$ with unit trace. Technically, the QKD protocol is therefore defined for the space $\mathcal{P}_=(\mathcal{H}_A \otimes \mathcal{H}_B)$, but it is sufficient to define it for $\mathcal{H}_A \otimes \mathcal{H}_B$ and to extend it by linearity.
In the case of a continuous-variable protocol, $\mathcal{H}$ corresponds to an $n$-mode Fock space: $\mathcal{H} \cong F^{\otimes n}$ where $F = \mathrm{Span}(|0\rangle, \ldots, |k\rangle, \ldots)$ and $|k\rangle$ is a $k$-photon Fock state. The fact that this space is infinite-dimensional makes the analysis much more involved than for BB84, already for the case of collective attacks where the relevant space is $F \otimes F$ in a CV protocol (i.e. it is sufficient to consider a single copy of the state and $\mathcal{H}_A \otimes \mathcal{H}_B$ is simply $F \otimes F$), compared to the 2-qubit space $\mathbbm{C}^2 \otimes \mathbbm{C}^2$ for BB84. 

The spaces $\mathcal{S}_A$, $\mathcal{S}_B$, $\mathcal{C}$ correspond to classical registers (with subscripts $A$ and $B$ referring to Alice and Bob, respectively). Indeed a QKD protocol takes a quantum state as its input and outputs classical strings. (An exception is device-independent QKD where the inputs are also classical but then, the violation of a Bell inequality ensures that Alice and Bob are indeed measuring an entangled quantum system.) The spaces $\mathcal{S}$ are that of the final keys and $\mathcal{C}$ is the public transcript of the protocol, corresponding to the classical information exchanged on the authenticated classical channel and therefore accessible by Eve.
For instance, the register $\mathcal{C}$ contains the size $l$ of the final key, which can be zero if the protocol aborted.

A QKD protocol should be \emph{secure}, which means that it should display both properties of \emph{secrecy} and \emph{correctness}. It should also be \emph{robust} in the sense that it should output nontrivial keys if there is no active attack on the quantum channel. Here, we follow the definitions of Ref.~\cite{TLG12}.
These properties correspond to properties of the output state, more precisely of $\rho_{S_A S_B E}$, which should hold for any input state. (Here, and in the remainder of the text, $\rho_H$ corresponds to the marginal of the state $\rho$ restricted to subspace $H$.) We note the appearance of the register $E$ corresponding to the Hilbert space $\mathcal{H}_E$ of the adversary. To formalize it, we imagine that the true input space of the protocol is $\mathcal{H}_A \otimes \mathcal{H}_B \otimes \mathcal{H}_E$, and that the input state is a pure state (this is without loss of generality). The protocol then acts as $\mathcal{E}_{AB} \otimes \mathrm{id}_E$, i.e. it acts trivially on the adversary' s space. The output space is $ \mathcal{S}_A \otimes \mathcal{S}_B \otimes \mathcal{T}_A \otimes \mathcal{T}_B \otimes \mathcal{C} \otimes \mathcal{H}_E$. Denote by $\mathcal{H}_{E'} = \mathcal{C} \otimes \mathcal{H}_E$ the space accessible to the adversary (corresponding to her own quantum system $E$ and to the totality of public communication $\mathcal{C}$). We will be interested in the state $\rho_{S_A S_B E'}$ describing the (output)  joint state of the keys of Alice and Bob and of the adversary. 
More precisely, a QKD protocol is called \emph{correct} if $S_A = S_B$ for any strategy of the adversary, that is, any initial state of the protocol $\Psi_{ABE}$. A protocol is $\epsilon_{\mathrm{cor}}$-correct if $\mathrm{Pr}[S_A \ne S_B] \leq \epsilon_{\mathrm{cor}}$.
A key is called $\delta$-secret if it is $\delta$-close to a uniformly distributed key that is uncorrelated with the eavesdropper:
\begin{align}
\frac{1}{2} \left\| \rho_{S_A^l  E'} - \omega_{l} \otimes \rho_{E'} \right\|_1 \leq \delta,
\end{align}
where $\rho_{S_A}^l$ is the state conditioned on the key length being $l$ and  $\omega_{l}$ is the fully mixed state on classical strings of length $l$.
We immediately note that if the protocol aborts, the corresponding key, is automatically secret since one party aborting the protocol immediately informs the other party about this decision. A QKD protocol is called $\epsilon_{\mathrm{sec}}$-secret if it outputs $\delta$-secret keys with $(1-p_{\mathrm{abort}}) \delta \leq \epsilon_{\mathrm{sec}}$.
Here, $p_{\mathrm{abort}}$ represents the probability that the protocol aborts. This probability depends on the strategy of the adversary, that is on the input state $\Psi_{ABE}$. Indeed, the adversary can always choose to cut the line between Alice and Bob and hence make sure that the protocol always aborts. This is fine since the key will nevertheless be secure. An important parameter is the \emph{robustness}, $\epsilon_{\mathrm{rob}}$, of the protocol, which corresponds to the abort probability if the adversary is passive, and if the characteristics of the quantum channel are conform to what was expected. For instance, in the case of a CV QKD protocol, a typical quantum channel corresponding to an optical fiber will be a Gaussian channel with fixed transmittance $T$ and excess noise $\xi$.
A QKD protocol is $\epsilon$-secure if it is $\epsilon_{\mathrm{sec}}$-secret and $\epsilon_{\mathrm{cor}}$-correct with $\epsilon_{\mathrm{sec}} + \epsilon_{\mathrm{cor}} \leq \epsilon$.

A generic technique to prove that a protocol is $\epsilon$-secure is to show that it is indistinguishable from an ideal protocol. Here, indistinguishability refers to two CPTP maps, and is quantified by the diamond distance between these maps, with the operational property that the maximum probability of correctly guessing whether a map corresponds to $\mathcal{E}$ or $\mathcal{F}$ is given by $p = \frac{1}{2} + \frac{1}{4} \|\mathcal{E} - \mathcal{F}\|_\diamond$. 
Here the map $\mathcal{E}$ will be the QKD protocol and $\mathcal{F}$ should be an ideal version of the same QKD protocol. This ideal version is for instance obtained by concatenating $\mathcal{E}$ and a (virtual) protocol $\mathcal{P}$ that replaces the final keys $S_A$ and $S_B$ by a perfect key $S$: $\mathcal{P}(\rho_{S_A S_B E'}^l) = \omega_{l,l} \otimes \rho_{E'}$, where $\omega_{l,l}$ is the totally mixed state on two copies of strings of length $l$. One defines $\mathcal{F} = \mathcal{P} \circ \mathcal{E}$ and the protocol $\mathcal{E}$ is $\epsilon$-secure if 
\begin{align}
\frac{1}{2} \|\mathcal{E}- \mathcal{F}\|_\diamond \leq \epsilon.
\end{align}

We note that for (finite-dimensional) \emph{symmetric} protocols, i.e. such that $\mathcal{E} \circ \pi = \mathcal{E}$ for any operation $\pi$ that permutes simultaneously Alice's and Bob's $n$ quantum subsystems, this diamond distance can be bounded by computing the distance when applied to an independent and identically distributed (i.i.d.) state. This is the idea behind the Postselection technique \cite{CKR09}, which shows that for such protocols, collective attacks are asymptotically optimal. A crucial question, however, is that of verifying that a protocol displays the required symmetry. This is the case for protocols such as BB84 with strict one-way classical communication (in fact, the protocol is first restricted to one where Alice prepares the initial entangled state, which cannot be given by the adversary), but in general, it might be necessary to actively symmetrize the state (or at least the raw keys) in order to make the protocol symmetric. We will come back to this question later since symmetrization is a costly process that one would like to avoid in a practical implementation. 

All these notions extend to Prepare-and-Measure (PM) protocols in the following way. A PM protocol and an EB protocol are equivalent if they are indistinguishable from any coalition of observers outside of Alice's lab, and if Alice outputs the same key in both cases. The security of a given PM protocol is a consequence of the security of the equivalent EB protocol. 

Note finally that in this paper, we make the assumption that the measurement devices of Alice and Bob are trusted and behave accordingly to their theoretical model. In the protocol we consider, the measurements correspond to heterodyne detection. For a single-mode state $\rho$, the probability density function of measurement outcome is $p(\alpha) = \langle \alpha |\rho|\alpha\rangle$ where $|\alpha\rangle = e^{-|\alpha|^2/2} \sum_{k=0}^\infty \frac{\alpha^k}{\sqrt{k!}}|k\rangle$ is a coherent state centered on $\alpha \in \mathbbm{C}$ in phase-space. In particular, these assumptions imply that Eve cannot tamper with the Local Oscillator used for the detection. To avoid too much complication, we also assume that Alice's state preparation and Bob's detection are ideal, but note that incorporating (trusted) imperfections to these models (for instance, imperfect quantum detection efficiency) is straightforward using similar techniques as in Ref.~\cite{LBG07} for instance.  
Moreover, we use the convention that the shot-noise variance is equal to 1.

\section{Description of the CV QKD protocol}
\label{protocol}

The main part of this manuscript is devoted to the analysis of a specific protocol, denoted by $\mathcal{E}_0$, for which we prove composable security against collective attacks. We focus on the EB version of the protocol here, with reverse reconciliation. The direct reconciliation version is easily derived by interchanging the roles of Alice and Bob in the classical post-processing part of the protocol. 
In order to obtain security against general attacks, one needs to add another step to the protocol, involving an energy test as well as a potential symmetrization procedure. This will be described in Section \ref{general-attacks}.

The protocol is characterized by a number of parameters:
\begin{itemize}
\item the number $2n$ of light pulses (coherent states) exchanged during the protocol. 
\item The size $l$ of the final key if the protocol did not abort. 
\item The number $d$ of bits on which each measurement result is encoded. For numerical applications, we will use $d=5$. In general, the value of $d$ should be optimized as a function of the specific reconciliation procedure used in the protocol. 
\item The size of Bob's communication to Alice, $\mathrm{leak}_{\mathrm{EC}}$, during the error correction procedure: this includes the syndrome of Bob's string for an error correcting code that was decided in advance by Alice and Bob, as well as the size of a small hash that will allow them to check that the error correction succeeded, except with some small probability $\epsilon_{\mathrm{cor}}$. Here failure means that the keys of Alice and Bob do not coincide and that the protocol did not abort. 
\item The number of bits $n_{\mathrm{PE}}$ that Bob sends to Alice during the Parameter Estimation.
\item Some bounds on covariance matrix elements, $\Sigma_a^{\max}, \Sigma_b^{\max}, \Sigma_c^{\min}$. The role of the Parameter Estimation procedure is to make sure that these bounds are sound in the realization of the protocol. Typically, these values are optimized as a function of the expected characteristics (transmittance and excess noise) of the quantum channel.
\item A maximum failure probability of the Parameter Estimation procedure, $\epsilon_{\mathrm{PE}}$.
\end{itemize}

A sketch of the protocol $\mathcal{E}_0$ is displayed on Fig.~\ref{protocol-E0-old}. 
\begin{figure}[htbp]
\begin{framed}
  \centering
\begin{enumerate}
\item \textbf{State Preparation:} Alice and Bob each have access to a $2n$-mode state. The global state is denoted by $\rho_{AB}^{2n}$.
\item \textbf{Measurement:} Alice and Bob measure their modes with heterodyne detection. They obtain two strings $X, Y \in \mathbbm{R}^{4n}$. Bob discretizes his $4n$-vector $Y$ in order to obtain the $m$-bit string $U$, where $m = 4dn$, i.e. each symbol is encoded with $d$ bits of precision.
\item \textbf{Error Correction:} Bob sends some side information to Alice (syndrome of $U$ for a linear error correcting code $C$ agreed on in advance) and Alice outputs a guess $U_A$ for the string of Bob. Bob computes a hash of $U$ of length $\lceil \log(1/\epsilon_{\mathrm{cor}}) \rceil$ and sends it to Alice who compares it with her own hash. If both hashes coincide, the protocol resumes, otherwise it aborts. The value $\mathrm{leak}_{\mathrm{EC}}$ corresponds to the total number of bits sent by Bob during the error correction phase.
\item \textbf{Parameter Estimation:} Bob sends $n_{\mathrm{PE}}$ bits of information to Alice that allow her to obtain three values $\gamma_a, \gamma_b$ and $\gamma_c$ defined in Eq.~\ref{def-gammaa-appendix}, \ref{def-gammab-appendix} and \ref{def-gammac-appendix}. 
If $\gamma_a \leq \Sigma_a^{\max}$ and $\gamma_b \leq \Sigma_b^{\max}$ and $\gamma_c \geq \Sigma_c^{\min}$, then the protocol continues. Otherwise it aborts. 
\item \textbf{Privacy Amplification:} Alice and Bob apply a random universal$_2$ hash function to their respective strings, obtaining two strings $S_A$ and $S_B$ of size $l$. 
\end{enumerate} 

\end{framed}
\caption{Protocol $\mathcal{E}_0$, with reverse reconciliation and parameters $n, l, d, \mathrm{leak}_{\mathrm{EC}}, \epsilon_{\mathrm{cor}}, n_{\mathrm{PE}}, \epsilon_{\mathrm{PE}}, \Sigma_a^{\max}, \Sigma_b^{\max}, \Sigma_c^{\min}$. }
\label{protocol-E0-old}
\end{figure}

We also recall our main result:
\begin{theo} 
\label{key-rate-theorem-2}
The protocol $\mathcal{E}_0$ is $\epsilon$-secure against collective attacks if $\epsilon = 2\epsilon_{\mathrm{sm}} + \bar{\epsilon} + \epsilon_{\mathrm{PE}}/\epsilon + \epsilon_{\mathrm{cor}}/\epsilon + \epsilon_{\mathrm{ent}}/\epsilon$ and 
\begin{align}
l \leq & 2n \left[ 2\hat{H}_{\mathrm{MLE}}(U) - f(\Sigma_a^{\max}, \Sigma_b^{\max}, \Sigma_c^{\min}) \right]- \mathrm{leak}_{\mathrm{EC}}  -  \Delta_{\mathrm{AEP}} - \Delta_{\mathrm{ent} } - 2\log\frac{1}{2\bar{\epsilon}},
\label{key-rate}
\end{align}
where $\hat{H}_{\mathrm{MLE}}(U)$ is the empirical entropy of $U$, $\Delta_{\mathrm{AEP}} := \sqrt{2n} \left[(d+1)^2 + 4(d+1) \log_2 \frac{2}{\epsilon_{\mathrm{sm}}^2} + 2\log_2 \frac{2}{\epsilon^2\epsilon_{\mathrm{sm}}}\right] - 4 \frac{\epsilon_{\mathrm{sm}} d}{\epsilon}$, $\Delta_{\mathrm{ent}} :=   \log_2 \frac{1}{\epsilon}-  \sqrt{8n \log_2^2 (4 n) \log(2/\epsilon)}$ and $f$ is the function computing the Holevo information between Eve and Bob's  measurement result for a Gaussian state with covariance matrix parametrized by $\Sigma_a^{\max}, \Sigma_b^{\max}, \Sigma_c^{\min}$.
\end{theo}
Note that in the main text, we used that $\epsilon =  \sqrt{\epsilon_{\mathrm{PE}} + \epsilon_{\mathrm{cor}} + \epsilon_{\mathrm{ent}}} + 2\epsilon_{\mathrm{sm}} + \bar{\epsilon} $ as a possible security parameter.

The function $f$ is defined as follows:
\begin{align}
f(x,y,z)  := g((\nu_1-1)/2)+g((\nu_2-1)/2)-g((\nu_3-1)/2),
\end{align}
where $\nu_1$ and $\nu_2$ are the symplectic eigenvalues of the covariance matrix $\left[
\begin{smallmatrix}
x \mathbbm{1}_2 & z \sigma_z \\
z  \sigma_z &y\mathbbm{1}_2 \\
\end{smallmatrix}
\right]$, $\nu_3 =x^2- \left(z^2 \right)^2/(1+ y)$, $\sigma_z = \mathrm{diag}(1,-1)$ and the entropy function $g$ is given by  $g(x) := (x+1) \log_2(x+1) -x \log_2 (x)$.

The symplectic eigenvalues $\nu_1$ and $\nu_2$ of the covariance matrix $\left[
\begin{smallmatrix}
x \mathbbm{1}_2 & z \sigma_z \\
z  \sigma_z &y\mathbbm{1}_2 \\
\end{smallmatrix}
\right]$
satisfy the following relations:
\begin{align}
\nu_1^2 + \nu_2^2 &= x^2 +y^2-2z^2\\
\nu_1^2 \nu_2^2 & = (xy-z^2)^2.
\end{align}
More explicitly,
\begin{align}
\nu_1 &= \sqrt{\frac{1}{2} \left( x^2 +y^2-2z^2  + \sqrt{(x^2 +y^2-2z^2)-2-4 (xy-z^2)^2}\right)}\\
\nu_2 &= \sqrt{\frac{1}{2} \left( x^2 +y^2-2z^2  - \sqrt{(x^2 +y^2-2z^2)-2-4 (xy-z^2)^2}\right)}.
\end{align}

The empirical entropy of $U$ is computed as follows. The random variable $U$ is discrete and takes values in a set of size $2^d$ (where $d$ is the discretization parameter). 
Let us denote by $\hat{n}_i$ the number of times the variable $U$ takes the value $i$, for $i \in \{1, \ldots, 2^d\}$, and let us denote by $\hat{p}_i = \frac{\hat{n}_i}{4n}$ the relative frequency of obtaining the value $i$. Note that there are $4n$ samples in total, since $2n$ states are measured, and each measurement gives two outcomes, one per quadrature. 
We define the Maximum Likelihood Estimator (MLE), also know as ``empirical entropy", for $H(U)$ to be 
\begin{align}
\hat{H}_{\mathrm{MLE}}(U):= - \sum_{i=1}^{2^d} \hat{p}_i \log \hat{p}_i.
\end{align}

Finally, the leakage term $\mathrm{leak}_{\mathrm{EC}}$ is simply the size of the syndrome transmitted by Bob to Alice during the error correction step. 

We now detail the various steps of the protocol. 

\subsection{State Preparation}
Alice prepares $2n$ copies of a two-mode squeezed vacuum state, $|\Phi\rangle^{\otimes 2n}$ where
$|\Phi\rangle = \left[ \frac{2}{V+1}\right]^{1/4} \sum_{k=0}^\infty \left[ \frac{V-1}{V+1}\right]^{k/2}|k,k\rangle$. The squeezing parameter of the state is optimized as a function of the expected characteristics of the quantum channel. The covariance matrix of the two-mode squeezed vacuum state (with the ordering convention $\hat{q}_A, \hat{p}_A, \hat{q}_B, \hat{p}_B$ where $\hat{q}_{A/B}$ and $\hat{p}_{A/B}$ are the quadrature operators for Alice and Bob, respectively) is given by
\begin{align}
\Gamma_{\mathrm{TMSS}} = \left[
\begin{matrix}
V \mathbbm{1}_2 & \sqrt{V^2-1} \sigma_z \\
\sqrt{V^2-1} \sigma_z & V \mathbbm{1}_2 \\
\end{matrix}
\right],
\end{align}
where $\mathbbm{1}_2 = \mathrm{diag}(1,1)$ and $\sigma_z = \mathrm{diag}(1,-1)$.
Alice keeps the first mode of each state $|\Phi\rangle$ and sends the second half to Bob through an insecure quantum channel. 
When proving security, we will not make any assumption about the quantum channel, which can be arbitrary. When considering the robustness of the protocol, however, it makes sense to model the channel as one that typically occurs in implementations. Our model of choice is a Gaussian channel with fixed transmittance $T$ and excess noise $\xi$. In this specific case, the quantum state shared by Alice and Bob is Gaussian and its covariance matrix reads:
\begin{align}
\Gamma_{\mathrm{Gauss}} = \left[
\begin{matrix}
V \mathbbm{1}_2 &\sqrt{T} \sqrt{V^2-1} \sigma_z \\
\sqrt{T} \sqrt{V^2-1} \sigma_z & (TV - T+1+T\xi) \mathbbm{1}_2 \\
\end{matrix}
\right].
\end{align}

\subsection{Measurement}

\subsubsection{Obtaining two real-valued vectors}

In the Entanglement-Based version of the protocol, the measurement phase is straightforward. Alice and Bob both have access to $2n$ modes, which they measure with a heterodyne detection. We denote by $\hat{q}$ and $\hat{p}$ the two quadratures for each mode. Alice and Bob then form two vectors of length $4n$ where odd index coordinates refer to quadrature measurement $\hat{q}$ and even index coordinates refer to quadrature $\hat{p}$. 
We denote these two vectors by $X^{\mathrm{exp}} = (X^{\mathrm{exp}}_1, \ldots, X^{\mathrm{exp}}_{4n})$ for Alice and $Y^{\mathrm{exp}} = (Y^{\mathrm{exp}}_1, \ldots, Y^{\mathrm{exp}}_{4n})$ for Bob. 
The superscript $\mathrm{exp}$ refers to experimental data. 

Conditioned on her measurement outcomes $X^{\mathrm{exp}}$, Alice has effectively prepared $2n$ coherent states for Bob (before the quantum channel), $\{|X^{\mathrm{PM}}_{2k+1} + i X^{\mathrm{PM}}_{2k+2}\rangle\}_{k=1\cdots 2n}$ where the vector $X^{\mathrm{PM}}$ is related to $X^{\mathrm{exp}}$ through 
\begin{align}
\label{PM-EB}
X_k^{\mathrm{PM}} := \left\{ 
\begin{array}{cc}
\sqrt{\frac{V-1}{V+1}} X_k^{\mathrm{exp}}  & \text{if $k$ is odd,}\\ 
 -\sqrt{\frac{V-1}{V+1}}  X_k^{\mathrm{exp}}  & \text{if $k$ is even.}
\end{array}
\right.
\end{align}
Here $V$ corresponds to the variance of the initial two-mode squeezed vacuum states prepared by Alice. 

The equivalence between the Entangled-Based and the Prepare-and-Measure versions of the protocol is seen from the fact that Alice's vectors in both cases are related through Eq.~\ref{PM-EB}. We recall that in the Prepare-and-Measure version, Alice would simply prepare the $2n$ coherent states $\{|X^{\mathrm{PM}}_{2k+1} + i X_{2k+2}^{\mathrm{PM}}\rangle\}_{k=1\cdots 2n}$ and send them to Bob, where the random variables $X_k^{\mathrm{PM}}$ are i.i.d.~centered Gaussian variables with variance $(V-1)$.

\subsubsection{Centering the measurement outcomes}

In the QKD protocols investigated in this paper, Alice and Bob will measure a quantum state with a heterodyne detection and obtain continuous-valued outcomes  $X_1^{\mathrm{exp}}, \ldots, X_{4n}^{\mathrm{exp}}$ and $Y_1^{\mathrm{exp}}, \ldots, Y_{4n}^{\mathrm{exp}}$ (where the components with an odd index correspond to measurement outcomes for the $\hat{q}$ quadrature and components with an even index correspond to $\hat{p}$ quadrature measurements), and where the superscript $\mathrm{exp}$ refers to the fact that the variables correspond to the effectively measured experimental data. 

The security of the key that can be extracted from a given state with the protocols described here does not depend on the first moment of the quantum state: only its covariance matrix matters \cite{GC06, NGA06}. For that reason, Alice and Bob can apply a displacement to their state before measuring it or alternatively simulate this displacement at the level of classical data. 
For reasons that will appear more clearly later (related to the symmetrization procedure needed for Parameter Estimation for instance), it is better for the quantum state to be centered. 
To do that, Alice and Bob can compute two values for the displacement: 
\begin{align}
D_q^A := \frac{1}{2n} \sum_{k=0}^{2n-1} X_{2k +1}^{\mathrm{exp}}, \quad D_p^A := \frac{1}{2n} \sum_{k=1}^{2n} X_{2k}^{\mathrm{exp}}, \quad
D_q^B := \frac{1}{2n} \sum_{k=0}^{2n-1} Y_{2k +1}^{\mathrm{exp}} \quad \text{and} \quad D_p^B := \frac{1}{2n} \sum_{k=1}^{2n} Y_{2k}^{\mathrm{exp}}
\end{align}
and define the new variables 
\begin{align}
X_k & := \left\{ 
\begin{array}{cc}
X_k^{\mathrm{exp}} - D_q^A & \text{if $k$ is odd,}\\ 
X_k^{\mathrm{exp}} - D_p^A & \text{if $k$ is even,}
\end{array}
\right.\\
Y_k & := \left\{ 
\begin{array}{cc}
Y_k^{\mathrm{exp}} - D_q^B & \text{if $k$ is odd,}\\ 
Y_k^{\mathrm{exp}} - D_p^B & \text{if $k$ is even,}
\end{array}
\right.
\end{align}
and use these variables in the remainder of the QKD protocol.
The resulting random variables $X_k$ and $Y_k$ are therefore centered. 
One advantage of working with centered variables is that the covariance matrix elements can be simply expressed in terms the expectations of $\|X\|^2, \|Y\|^2$ and $\langle X,Y\rangle$.

\subsubsection{Discretization of a continuous variable}

In this work, we only consider discretization of continuous variables, one at a time. More precisely, we will divide the real axis into $2^d$ intervals. Here, the choice of a power of 2 is only made for convenience. 
This partition should be chosen so as to maximize the secret key rate in the case where the adversary is passive, when the quantum channel acts as a Gaussian channel with fixed transmittance and excess noise. In that case, Bob's outcomes are centered Gaussian random variables. 
Bob can easily compute the average variance of his measurement outcomes, $\frac{1}{4n} \|Y\|^2$. 
Consider the $2^d$ quantiles $\mathcal{I}_1, \ldots, \mathcal{I}_{2^d}$ of the normal distribution $\mathcal{N}(0, \frac{1}{4n} \|Y\|^2)$. We will apply the discretization map $\mathcal{D}: Y \mapsto U$ that assigns a distinct value for each quantile: $\mathcal{D}(Y_k) = j$ if $Y_k \in \mathcal{I}_j$.

We note that if the quantum channel is indeed Gaussian, then the random variable $U = \mathcal{D}(Y)$ should be close to uniform, i.e. have almost maximum entropy.
In fact, this discretization is suboptimal because it does not maximize the mutual information between Alice and Bob, $I(X;U)$. This will lead to a slightly suboptimal secret key rate, but with the advantage of a simpler security analysis. Better schemes involving a multidimensional error correction have been investigated in the literature \cite{LAB08} and could probably be proven secure against general attacks. 

At the end of this step, Alice knows the vector $X\in \mathbbm{R}^{4n}$ and Bob knows both vectors $Y \in \mathbbm{R}^{4n}$ and $U \in \{1, \ldots, 2^d\}^{4n}$.

\subsection{Error correction}

The goal of this step is for Alice to learn the string $U$. 
In order to help her calibrate her data, Bob first sends her the value of $\|Y\|^2$ which is used for the discretization function $\mathcal{D}$. This can be done by sending the value of $\sqrt{\|Y\|^2/(4n)}$ with a few bits of precision.
A possible technique to perform error correction is called \emph{slice reconciliation} \cite{VCC04, JEK14}. This means that Alice and Bob agree (before they start the QKD protocol) on a linear error correcting code $C$ (or, more precisely, on a family of such codes) that encodes $K$-bit binary strings into $4dn$-bit code words. This code is described by a parity-check matrix $H$ of size $(4dn) \times (4dn-K)$. 
Bob computes the syndrome $HU$ of his vector (interpreted as a binary string of length $4dn$) and sends this syndrome to Alice. This syndrome accounts for most of the leakage during the error correction procedure.
A figure of merit typically used in CV QKD to assess the quality of the error correction is the \emph{reconciliation efficiency} $\beta$ defined as:
\begin{align}
\beta = \frac{4dn - \mathrm{leak}_{\mathrm{EC}}}{2n \log_2(1+\mathrm{SNR})}, 
\end{align}
where the signal-to-noise ratio (SNR) is that of the expected Gaussian channel mapping $X$ to $Y$. The reconciliation efficiency compares how much information was extracted through the error correction procedure and compares it with the available mutual information corresponding to a Gaussian quantum channel of transmittance $T$ and excess noise $\xi$. In particular, one expects:
\begin{align}
\mathrm{SNR} = \frac{T(V-1)}{2+T\xi}.
\end{align}
In that case, the quantity $\frac{1}{2} \log_2(1+\mathrm{SNR})$ is the mutual information for each use of the channel, and gives an upper bound on the quantity of (classical) information that can be transmitted over the classical channel $X_i \mapsto Y_i$. 
The reconciliation efficiency therefore quantifies how far the error correction procedure is from an ideal one: if $\beta = 1$, then it is perfect. In practice, one can typically achieve $\beta \approx 0.95$ for a Gaussian channel \cite{JEK14}.

Once Alice learns the syndrome $HU$, she can use it together with her vector $X$ and the value of $\|Y\|$ in order to recover an estimate $\hat{U}$ of $U$. This is done by decoding the code $C$ in the coset corresponding to the syndrome. There are very practical algorithms for decoding as soon as the code $C$ is a low-density parity-check code for instance.

After this step, it is necessary to know whether the error correction worked, i.e. whether $\hat{U} = U$ or not. To achieve this, the usual technique is for Alice and Bob to choose a random universal$_2$ hash function mapping $4dn$-bit strings to strings of length $\lceil \log(1/\epsilon_{\mathrm{cor}}) \rceil$. Then Bob reveals his hash to Alice. If both hashes coincide, the protocol resumes, otherwise it aborts.
The length of the hash is chosen so the protocol is $\epsilon_{\mathrm{cor}}$-correct (see Theorem $1$ of Ref.~\cite{TLG12} for a proof of this statement).

\subsection{Parameter Estimation}

The Parameter Estimation is arguably the most crucial step in a QKD protocol: it is a coarse-grained version of quantum tomography, the process of inferring the description of a quantum state when one only has access to measurement outcomes. 
It is only a coarse-grained version of tomography because the security of the key usually only depends on a small number of parameters of the underlying quantum state. 

Before we describe in detail the Parameter Estimation step of the protocol $\mathcal{E}_0$, it is useful to explain a bit more how one can perform the tomography of continuous-variable quantum systems. Our approach is based on the paradigm introduced by Christandl and Renner in Ref.~\cite{CR12}.

\subsubsection{Estimation of the covariance matrix of a CV system}

One specific aspect of quantum tomography is that it does not make sense to speak of a quantum state if the state does not exist anymore, because it was already measured, say. 
For this reason, a quantum tomography process can be described as a CPTP map as follows:
\begin{align}
 \begin{array}{cccc}
\mathcal{T}\!\mathrm{om}\colon & \mathcal{H}^{\otimes (n+k)} & \rightarrow & \mathcal{H}^{\otimes n} \otimes \mathcal{R}\\
& \rho^{n+k} & \mapsto & \rho^n \otimes R,
\end{array}
\end{align}
where $R$ is a classical random variable corresponding to a \emph{confidence region} for the output state $\rho^n$. More precisely, the random variable $R$ describes a region of $\mathcal{P}_{=}(\mathcal{H}^{\otimes n})$, the set of normalized density operators on $\mathcal{H}^{\otimes n}$, which is believed to contain the state $\rho^n$. Here the superscript $n$ for $\rho^n$ should not be interpreted as saying that the state has an i.i.d.~structure, nor that $\rho^n$ corresponds to a marginal state of $\rho^{n+k}$; it is merely a reminder of the size of Hilbert space it lives in. 

The quantum tomography protocol involves performing measurements on a subsystem (here living in $\mathcal{H}^{\otimes k}$) of the initial quantum state and the quality of the protocol is assessed by two parameters: the probability $\epsilon_{\mathrm{tom}}$ that the prediction is false and the size of the region. Of course, if the region $R$ is equal to the total space, $\mathcal{H}^{\otimes n}$, the prediction is always correct, but it is also uninformative. 
For a QKD protocol, the size of the region will influence the tightness of the bound on the secret key rate: smaller confidence regions lead to tighter key rates. But the more crucial parameter is the probability that the prediction is incorrect. 
This probability is the smallest value of $\epsilon_{\mathrm{tom}}$ such that:
\begin{align}
\forall \rho^{n+k} \in\mathcal{H}^{\otimes (n+k)}, \mathrm{Pr}\left[ \rho^n \in R \right] \geq 1-\epsilon_{\mathrm{tom}}
\end{align}
where $\rho^n$ and $R$ are the output of the map $\mathcal{T}\!\mathrm{om}$ applied to $\rho^{n+k}$.  

In the context of CV QKD, we are interested in estimating the (averaged) covariance matrix of a bipartite state $\tilde{\rho}_{AB}^{n}$ (this state is in fact obtained from $\rho_{AB}^{2n}$ after a suitable symmetrization as we will explain later) which is characterized by three real values:
\begin{align}
\Sigma_a & := \frac{1}{2n} \sum_{i=1}^n \left[ \langle q_{A_i}^2 \rangle + \langle p_{A,i}^2\rangle \right]\\
\Sigma_b & := \frac{1}{2n}  \sum_{i=1}^n  \left[\langle q_{B_i}^2 \rangle + \langle p_{B,i}^2\rangle \right]\\
\Sigma_c & := \frac{1}{2n}  \sum_{i=1}^n  \left[\langle q_{A,i} q_{B,i} \rangle - \langle p_{A,i} p_{B,i}\rangle \right]
\label{cov-matrix-elements}
\end{align}
where $q_{A,i}$ is the quadrature operator $\frac{1}{\sqrt{2}} (\hat{a}_i + \hat{a}^\dagger_i)$ for the $i^{\mathrm{th}}$ mode of Alice for instance. 
For CV QKD, the region $R$ will be a confidence region for the three parameters $\Sigma_a, \Sigma_b$ and $\Sigma_c$. 
There are two reasons for this choice. The first one is the extremality property of Gaussian states for the Holevo information\cite{WGC06,GC06}: this says that this information can be upper bounded by a function of the covariance matrix of the state. Second, it is sufficient to compute this bound for the following symmetrized covariance matrix:
\begin{align}
\Gamma^{\mathrm{sym}} := \bigoplus_{i=1}^n \left[
\begin{matrix}
\Sigma_a & 0 & \Sigma_c & * \\
0 & \Sigma_a & * & -\Sigma_c \\
\Sigma_c & * & \Sigma_b & 0 \\
* & -\Sigma_c & 0 & \Sigma_b \\
\end{matrix}
\right].
\end{align} 
where the entries $*$ are not specified. More precisely, one can always assume that these entries are 0.
This will be detailed in the section devoted to the security proof. 
Moreover, because of the properties of the Holevo information, it will be sufficient for our purpose to obtain a confidence region for $(\Sigma_a, \Sigma_b, \Sigma_c)$ of the form $[0, \Sigma_a^{\max}] \times [0, \Sigma_b^{\max}] \times [\Sigma_c^{\min}, \infty]$.

Usually the task of quantum tomography is greatly simplified if we only require it to hold for initial states which are i.i.d., that if of the form $\rho^{\otimes (n+k)}$. When it comes to estimating a covariance matrix, however, this assumption still appears quite weak. The problem comes from the fact that the coefficients of the covariance matrix are \textit{a priori} unbounded (this is a fundamental difference with finite-dimensional systems such as qubits, for which the density matrix elements are obviously bounded). 
Consider for instance the state $\sigma^{\otimes (n+k)}$ with $\sigma = (1-\epsilon) |0\rangle \langle 0| + \epsilon |N\rangle \langle N|$. The covariance matrix of $\sigma$ is $\mathrm{diag}(1+ N\epsilon/2, 1+N\epsilon/2)$ but any tomographic procedure that only examines $k \ll 1/\epsilon$ modes will conclude that the covariance matrix is close to that of the vacuum, which is clearly false if $N\epsilon \gg1$.
In order to solve this issue, the solution consists in first appropriately symmetrizing the state $\rho^{n+k}$ before measuring $k$ subsystems and inferring properties for the remaining $n$ modes. 
This will solve the problem of parameter estimation, but create a new issue, namely that the output state of the tomographic procedure will not have a i.i.d.~structure anymore, making the analysis of collective attacks more complicated.

The symmetrization that makes sense for CV QKD is the one that maximally symmetrizes the state while leaving the (averaged) values $\Sigma_a, \Sigma_b, \Sigma_c$ of the covariance matrix unchanged. 
We explain this symmetrization now for an initial state $\rho^{n+k}$ on $F_A^{\otimes (n+k)} \otimes F_B^{\otimes (n+k)}$.
Consider the unitary group $U(n+k)$ acting on $\mathbbm{C}^{n+k}$. This group acts in a natural way on the $(n+k)$-mode Fock space $F^{\otimes(n+k)}$ by associating to $V \in U(n+k)$ the unitary operator $\Phi(V)$ acting on $F^{\otimes(n+k)}$ which maps the vector $\vec{a} = (a_1, \ldots, a_{n+k})$ of annihilation operators of the $n+k$ modes to the vector $V\vec{a}$. This operation corresponds to a passive linear symplectic map in phase space and can be implemented thanks to a network of beamsplitters and phase-shifts acting on the $n+k$ modes. 
The group $U(n+k)$ also acts naturally on $F_A^{\otimes (n+k)} \otimes F_B^{\otimes (n+k)}$ by associating to $V$ the unitary $\Phi(V)_A \otimes \Phi(V^*)_B$ where $V^*$ represents the complex conjugate of $V$. 
The symmetrization procedure consists in drawing a random unitary $V$ from the Haar measure $dV$ on $U(n+k)$ and applying $\Phi(V)_A \otimes \Phi(V^*)_B$ to the state:
\begin{align}
\label{sym-map}
\begin{array}{cccc}
\mathrm{Sym}\colon &  \mathcal{P}_=\left(F_A^{\otimes (n+k)} \otimes F_B^{\otimes (n+k)}\right) & \rightarrow & \mathcal{P}_=\left(F_A^{\otimes (n+k)} \otimes F_B^{\otimes (n+k)}\right) \\
& \rho_{AB}^{n+k} & \mapsto & \int \left(\Phi(V)_A \otimes \Phi(V^*)_B\right) \rho_{AB}^{n+k} \left(\Phi(V)_A \otimes \Phi(V^*)_B\right)^\dagger \mathrm{d}V.
\end{array}
\end{align}

The Parameter Estimation procedure that we have in mind first symmetrizes the state with $\mathrm{Sym}$, then measures the last $k$ modes with heterodyne detection and uses the measurement outcomes to give a confidence region for $\Sigma_a, \Sigma_b, \Sigma_c$.
In the QKD protocol, we will in fact choose $k=n$. If Alice and Bob were to proceed with the tomography procedure explained above, they would therefore measure $n$ modes each with heterodyne detection, and attempt to infer a confidence region for the covariance matrix of the $n$ remaining modes. 
Let us denote by $X_2$ and $Y_2$ the $2n$-vectors corresponding to their measurement results. 
In order to estimate $\Sigma_a, \Sigma_b, \Sigma_c$, Alice and Bob need to compute $\|X_2\|^2, \|Y_2\|^2$ and $\langle X_2, Y_2\rangle$. Since we study a QKD protocol, it is essential to limit classical communication to a minimum, because it might help the adversary. 
Our goal is for Alice to be able to perform the parameter estimation: it is clear that she can compute $\|X_2\|^2$. Similarly, Bob can compute $\|Y_2\|^2$ locally and simply send the result to Alice using a few bits to encode its value (as we explained before). 
Computing the inner product $\langle X_2, Y_2\rangle$ is a bit more problematic. In fact, computing the inner product of two vectors held by distant parties is a very well-known problem in communication complexity where the goal is to achieve the task while exchanging as little public information as possible. 
In the QKD setting, we are helped because the parameter estimation can be performed \emph{after} the error correction. This means in particular that Alice knows $X_2$ as well as $\hat{Y_2}$, her estimate of $Y_2$. Recall that this is only an estimate because the error correction procedure only allows her to recover Bob's discretized vector (with high probability if the protocol did not abort). 

Our goal is to obtain a lower bound on the value of $\langle X_2, Y_2 \rangle$, which will translate into a lower bound for $\Sigma_c$.

The discretization procedure maps $Y$ to $U$. We define another map that attempts to invert the discretization. Recall that $\mathcal{I}_i$ is the $i^{\mathrm{th}}$ quantile of the distribution $\mathcal{N}(0,v)$ for a specific variance $v$. One can define $2^d$ values, one for each quantile, as follows:
\begin{align}
\hat{y}_i := 2^d \int_{\mathcal{I}_i} \frac{x}{\sqrt{2\pi v}}e^{-\frac{x^2}{2v}} \mathrm{d}x,
\end{align}
which means that $\hat{y}_i$ corresponds to the mean of the Gaussian random variable conditioned on the fact that it is in the quantile $\mathcal{I}_i$.

Alice knows both $X_2$ and $\hat{Y}_2$ while Bob knows $Y_2$ and $\hat{Y}_2$. Since $\langle X_2, Y_2\rangle = \langle X_2, \hat{Y}_2\rangle + \langle X_2, Y_2-\hat{Y}_2\rangle$, it is sufficient for Alice to be able to estimate $\langle X_2, Y_2-\hat{Y}_2\rangle$, a quantity expected to be very small in practice. 
In order to achieve this, Bob first sends the norm of $\|Y_2-\hat{Y}_2\|$ to Alice. This requires only a small constant number of bits. 
Then, the communication problem that Alice and Bob should solve is the following: Alice knows a $2n$-dimensional unit vector $\vec{a} = \frac{X_2}{\|X_2\|}$, and Bob knows another unit vector $\vec{b} = \frac{Y_2-\hat{Y}_2}{\|Y_2-\hat{Y}_2\|}$, and they wish to estimate $\langle \vec{a}, \vec{b}\rangle$ up to a small additive error. 

This problem is studied in Ref.~\cite{KNR95}. The technique described by Kremer, Nisan and Ron gives the value $\langle \vec{a}, \vec{b}\rangle$ with additive error $\epsilon_2$ (except with probability $\epsilon_1$) with a $k$-round protocol with $ k = \Theta(\log(1/\epsilon_1)/\epsilon_2^2)$.
This means that Bob only needs to send approximately $\log(1/\epsilon_1)/\epsilon_2^2$ bits to Alice, so that she can compute $\langle X_2, Y_2\rangle$ arbitrary well.

In the remainder of this paper, we assume for simplicity that this task can be completed perfectly. This is a legitimate assumption since $\hat{Y}_2$ only differs from $Y_2$ because of discretization errors which are arbitrarily small. 
A thorough analysis of this type of errors and of the related optimization of the discretization step is left for future work. Note that there are no conceptual difficulties hidden here, but that optimizing this task only makes sense only if it is done jointly with the error correction procedure.

\subsubsection{Parameter Estimation in the protocol $\mathcal{E}_0$}

Now that we have explained the spirit of the Parameter Estimation procedure, we can introduce the PE test. 
Once Alice learns the values of $\|X\|^2, \|Y\|^2$ and $\langle X,Y\rangle$, she can compute:
\begin{align}
\gamma_a &:= \frac{1}{2n} \left[ 1 + 2\sqrt{\frac{\log(36/\epsilon_{\mathrm{PE}})}{n}}\right] \|X\|^2-1, \label{def-gammaa-appendix}\\
\gamma_b &:= \frac{1}{2n} \left[ 1 + 2\sqrt{\frac{\log(36/\epsilon_{\mathrm{PE}})}{n}}\right] \|Y\|^2-1,\label{def-gammab-appendix}\\
\gamma_c &:= \frac{1}{2n} \langle X, Y\rangle - 5 \sqrt{\frac{\log (8/\epsilon_{\mathrm{PE}})}{n^3}}(\|X\|^2 + \|Y\|^2).\label{def-gammac-appendix}
\end{align}
Then, she compares these values with the parameters $\Sigma_a^{\max}, \Sigma_b^{\max}$ and $\Sigma_c^{\min}$ of the protocol. If the three following conditions $\gamma_a \leq \Sigma_a^{\max}$ and $\gamma_b \leq \Sigma_b^{\max}$ and $\gamma_c \geq \sigma_c^{\min}$, then the protocol continues. Otherwise it aborts. The Parameter Estimation test is therefore defined as follows:

\textbf{Parameter Estimation Test:}
\begin{itemize}
\item \emph{Parameters}:  $\Sigma_a^{\max},  \Sigma_b^{\max},  \Sigma_c^{\min}$,
\item \emph{Input}: $\|X\|^2, \|Y\|^2, \langle X, Y\rangle$,
\item \emph{Output}: ``Pass'' if $\left[\gamma_a \leq \Sigma_a^{\max}\right] \wedge \left[\gamma_b \leq \Sigma_b^{\max} \right] \wedge \left[\gamma_c \geq \sigma_c^{\min}\right]$; ``Fail'' otherwise.
\end{itemize}

When fixing the parameters of the test, one should always apply a trade-off between the expected secret key rate and the robustness. 
Typically, if the variance of Alice's initial state is $V$, if the expected transmittance of the quantum channel is $T$ and the expected excess noise is $\xi$, one should choose 
\begin{align}
\Sigma_a^{\max} &= V + \delta_a\\
\Sigma_b^{\max} &= T(V-1)+1 +T\xi + \delta_b\\
\Sigma_c^{\min} &= \sqrt{T} (V-1) - \delta_c
\end{align}
where $\delta_a, \delta_b$ and $\delta_c$ are small positive constants which are optimized (as a function of $n$) to ensure both robustness and large secret key rate. 

In numerical applications, we choose a value of the robustness of about 1 percent, which is obtained for instance by choosing  $\delta_a, \delta_b$ and $\delta_c$ equal to 3 standard deviations for $\gamma_a, \gamma_b$ and $\gamma_c$ (for an expected Gaussian channel with transmittance $T$ and excess noise $\xi$).

\subsection{Privacy Amplification}

This step is completely standard: Alice chooses a universal$_2$ hash function \cite{BBC95,RK05} and extracts $l$ bits of secret $S_A$ from $\hat{U}$. She communicates the choice of function to Bob who uses it to compute $S_B$.

\section{Expected secret key rate}
\label{expected}

In order to compute the expected secret key rate provided by a security proof, one needs to model the quantum channel. Here, we will model it as a Gaussian channel with fixed transmissivity $T$ and fixed excess noise $\xi$. 

Our goal in this section is to explain how to reproduce the plot of Fig.~2 in the main text. 
Recall that the secret key rate is given by
\begin{align}
l \leq & 2n \left[ 2\hat{H}_{\mathrm{MLE}}(U) - f(\Sigma_a^{\max}, \Sigma_b^{\max}, \Sigma_c^{\min}) \right]- \mathrm{leak}_{\mathrm{EC}}  -  \Delta_{\mathrm{AEP}} - \Delta_{\mathrm{ent} } - 2\log\frac{1}{2\bar{\epsilon}},
\end{align}
with $\epsilon = 2\epsilon_{\mathrm{sm}} + \bar{\epsilon} + \epsilon_{\mathrm{PE}}/\epsilon + \epsilon_{\mathrm{cor}}/\epsilon + \epsilon_{\mathrm{ent}}/\epsilon$.

For concreteness, we wish to compute a secret key rate with $\epsilon = 10^{-20}$. In general, one should optimize over all the values of the parameters compatible with such an $\epsilon$, but here, we make the following (slightly suboptimal) choice:
\begin{align}
 \epsilon_{\mathrm{sm}} = \bar{\epsilon} = 10^{-21}, \quad \epsilon_{\mathrm{PE}} = \epsilon_{\mathrm{cor}} = \epsilon_{\mathrm{ent}}= 10^{-41}.
\end{align}

Moreover, we use the following model for the error correction:
\begin{align}
\beta I(A;B) = 2\hat{H}_{\mathrm{MLE}}(U)  - \frac{1}{2n} \mathrm{leak}_{\mathrm{EC}}
\end{align}
where $\beta$ is the so-called ``reconciliation efficiency" and $I(A;B)$ is the mutual information between Alice and Bob's classical data. 
Efficient error correction protocols are known for reconciling correlated Gaussian random variables and e choose $\beta = 0.95$ which is consistent with the best schemes available in the literature. 

For the Gaussian channel we consider, and if the variance modulation is $V$, we obtain:
\begin{align}
I(A;B) &= 2 \times \frac{1}{2} \log_2(1 +\mathrm{SNR}) \\
&= \log_2\left(1 + \frac{T(V-1)}{2/T\xi} \right).
\end{align}

Moreover, we choose the robustness of the protocol to be $\epsilon_{rob} \leq 10^{-2}$, which is obtained if the probability of passing the Parameter Estimation test is at least $0.99$. 
This can be achieved by taking values for $\Sigma_a^{\max}, \Sigma_b^{\max}, \Sigma_c^{\min}$ differing by 3 standard deviations from the expected values of $\gamma_a, \gamma_b, \gamma_c$.
With probability at least $0.99$, the values of random variables $\|X\|^2, \|Y\|^2, \langle X,Y\rangle$ satisfy the following inequalities:
\begin{align}
\|X\|^2 &\leq 2n(V+1) +3 \sqrt{4n(V+1)}\\ 
\|Y\|^2 &\leq 2 n (T (V - 1) + T \xi + 2) + 3 \sqrt{4 n (T (V - 1) + T \xi + 2)}\\
\langle X,Y\rangle & \geq 2 n \sqrt{T (V^2 - 1)} - 3 \sqrt{n (V - 1) (2 + T \xi)}
\end{align}
where we modeled each $(x_i, y_i)$ as identical and independent normal random variables, centered and with covariance matrix $\left[\begin{smallmatrix} V+1 & \sqrt{T (V^2 - 1)}  \\ \sqrt{T (V^2 - 1)} & T (V - 1) + T \xi + 2 \end{smallmatrix}\right]$. 
Finally, we use these bounds on $\|X\|^2, \|Y\|^2, \langle X,Y\rangle$ to define:
\begin{align}
\Sigma_a^{\max} &= \frac{1}{2n} \left[ 1 + 2\sqrt{\frac{\log(36/\epsilon_{\mathrm{PE}})}{n}}\right] \|X\|^2-1\\
\Sigma_b^{\max} &= \frac{1}{2n} \left[ 1 + 2\sqrt{\frac{\log(36/\epsilon_{\mathrm{PE}})}{n}}\right] \|Y\|^2-1\\
\Sigma_c^{\min} &= \frac{1}{2n} \langle X, Y\rangle - 5 \sqrt{\frac{\log (8/\epsilon_{\mathrm{PE}})}{n^3}}(\|X\|^2 + \|Y\|^2).
\end{align}

With all this, we are now in a position to compute the expected secret key rate displayed on Fig.~2 of the main text:
\begin{align}
r &= (1-\epsilon_{\mathrm{rob}})\frac{l}{2n}\\
&=  (1-\epsilon_{\mathrm{rob}}) \left(\beta \log_2\left(1 + \frac{T(V-1)}{2/T\xi} \right) - f(\Sigma_a^{\max}, \Sigma_b^{\max}, \Sigma_c^{\min})- \frac{1}{2n}\left[ \Delta_{\mathrm{AEP}} - \Delta_{\mathrm{ent} } - 2\log\frac{1}{2\bar{\epsilon}}\right]\right).
\end{align}
The last step is to optimize over values of the modulation variance $V$. 

\section{Tools for proving the security of the protocol $\mathcal{E}_0$ against collective attacks}
\label{tools}

In this section, we describe the various elementary tools that will be used in the security proof. First, the leftover hash lemma relates the secrecy of the protocol with the smooth min-entropy of the raw key $U$, conditioned on the adversary's system. Then, when the analysis is restricted to collective attacks, the Asymptotic Equipartition Property allows one to compute this smooth min-entropy as a function of the von Neumann entropy of a single subsystem. There are two problems remaining then: computing the entropy of the string $U$ and estimating the average covariance matrix of the quantum system. The first question is addressed thanks to a concentration result for the entropy of an i.i.d.~random variable, the second is taken care of with the analysis of the parameter estimation test.

\subsection{Leftover Hash Lemma}
One wishes to prove a statement about the secrecy of a given QKD protocol, i.e. that $\frac{1}{2} \left\| \rho_{S_A^l E'} - \omega_{l} \otimes \rho_{E'} \right\|_1 \leq \epsilon_{\mathrm{sec}}$. In other words, the variable $S_A$ should be decoupled from the system $E'$ characterizing what is available to the adversary, namely her own Hilbert space $\mathcal{H}_E$ as well as the public information, $\mathcal{C}$,  leaked during the QKD protocol.
Such a bound can be established via the \emph{Leftover Hash Lemma} \cite{ren05,TSS11}, which holds if $S_A$ is obtained by applying a random universal$_2$ hash function of length $l$ to the string $U$, where $l$ should be slightly smaller than the smooth min-entropy of $U$ conditioned on $E'$. The smooth min-entropy, $H_{\min}^\epsilon(U|E)$, introduced in Ref.~\cite{ren05}, characterizes the average probability that Eve guesses $U$ correctly using her optimal strategy with access to the correlations stored in her quantum memory \cite{KRS09}. For a precise mathematical definition, we refer the reader to Ref.~\cite{TCR09}. 

 More precisely, the Privacy Amplification procedure applied to the string $U$ outputs a the key of size $l$ which is $\epsilon_{\mathrm{sec}}$-secret provided that \cite{TLG12,TSS11,BFS11}
\begin{align}
\label{privacy-amp}
\epsilon_{\mathrm{sec}} = \min_{\epsilon'} \frac{1}{2} \sqrt{2^{l- H_\mathrm{min}^{\epsilon'}(U|E')}}+2\epsilon',
\end{align}
where $E'$ summarizes all the information Eve learned about $U$ during the protocol.

\subsection{Smooth min-entropy of a conditional state}

Evaluating the smooth min-entropy is usually a intractable optimization problem. Fortunately, in many interesting situations (such as the study of collective attacks), it is sufficient to evaluate it for i.i.d. states, in which case the Asymptotic Equipartition Property applies, and provides a bound expressed in terms of the von Neumann entropy. 

In the case of CV QKD, one needs to compute the smooth min-entropy of the state provided the protocol did not abort. Unfortunately, this postselection destroys the i.i.d. structure of the initial state, and the AEP does not directly apply anymore. In the following, we show that we can still obtain a weak version of the AEP for the postselected state and relate the smooth min-entropy to the von Neumann entropy. 

In particular, we show the following result. 

\begin{theo}[AEP for conditional state]
Let $\rho_{XB}$ be a classical quantum state, and denote by $d = \log_2 \mathrm{dim} \, \mathcal{H}_X$ so that the variable $X$ has cardinality $2^d$. 
Let $\tau_{X^nB^n} = \frac{1}{p} \Pi (\rho_{XB})^{\otimes n} \Pi$ with $p = \mathrm{tr}\, (\Pi (\rho_{XB})^{\otimes n})$ and $\Pi$ be any projector such that $\tau_{X^n B^n}$ and $\rho_{XB}^{\otimes n}$ commute. Then,
\begin{align}
H_{\mathrm{min}}^\epsilon(X^n|B^n)_{\tau_{X^n B^n}}  \geq H(X^n|B^n)_{{\tau}_{X^n B^n}} - \sqrt{n} \left[(d+1)^2 + 4(d+1) \log_2 \frac{2}{\epsilon^2} + 2\log_2 \frac{2}{p^2\epsilon}\right] - 4 \frac{\epsilon d}{p}.
\end{align}
\end{theo}
The proof uses many results from Chapter 6 of Tomamichel's thesis \cite{tom12}, and is partially based on discussions with Marco Tomamichel. 
\begin{proof}
Let us fix $\alpha = 1 + \frac{1}{\sqrt{n}}$. Prop $6.2$ from \cite{tom12} gives:
\begin{align}
H_{\mathrm{min}}^\epsilon(X^n|B^n)_{\tau}  \geq H_\alpha(X^n|B^n)_\tau -\sqrt{n} \log_2 \frac{2}{\epsilon^2}.
\end{align}
We now bound this quantity in terms of the $\alpha$-Rényi entropy computed for the i.i.d. state $\rho^{\otimes n}$:
\begin{align}
H_\alpha(X^n|B^n)_\tau & = \max_{\sigma_B^n} \frac{1}{1-\alpha} \log_2 \mathrm{tr}\, (\tau_{X^nB^n}^\alpha \sigma_{B^n}^{1-\alpha})\\
&=  \max_{\sigma_B^n} \frac{1}{1-\alpha} \log_2 \mathrm{tr}\, ((\Pi (\rho_{XB})^{\otimes n} \Pi)^\alpha \sigma_{B^n}^{1-\alpha}) -\frac{\alpha}{1-\alpha} \log_2 \frac{1}{p}\\
&\geq \max_{\sigma_B^n} \frac{1}{1-\alpha} \log_2 \mathrm{tr}\, ((\rho_{XB})^{\otimes n})^\alpha \sigma_{B^n}^{1-\alpha}) -\frac{\alpha}{1-\alpha} \log_2 \frac{1}{p}\\
&= H_\alpha(X^n|B^n)_{\rho^{\otimes n}} -\frac{\alpha}{1-\alpha} \log_2 \frac{1}{p}\\
&= n H_\alpha(X|B)_{\rho} -2\sqrt{n} \log_2 \frac{1}{p}.
\end{align}
Lemma $6.3$ from \cite{tom12} then yields:
\begin{align}
H_\alpha(X|B)_{\rho} \geq H(X|B)_\rho - \frac{4}{\sqrt{n}}(\log_2 \nu)^2
\end{align}
where $\nu :=  \sqrt{2^{-H_{\min}(X|B)_\rho}} + \sqrt{2^{H_{\max}(X|B)_\rho}} + 1 \leq 2^{d/2}+2 \leq 2^{(d+1)/2}-1/4$, for $d \geq 3$.  
So far, we have established that:
\begin{align}
\label{partial-proof}
H_{\mathrm{min}}^\epsilon(X^n|B^n)_{\tau}  \geq H(X^n|B^n)_{\rho^{\otimes n}} - \sqrt{n} \left[(d+1)^2 + 2\log_2 \frac{1}{p\epsilon}\right].
\end{align}
Intuitively, both conditional entropies $H(X^n|B^n)_{\tau}$ and $H(X^n|B^n)_{\rho^{\otimes n}}$ should be close, provided $p$ is not too small. Let us make this rigorous. 

Consider a purification $\rho_{XBE}$ of $\rho_{XB}$. Clearly, $\tau_{X^n B^n E^n} =  \frac{1}{p} (\Pi \otimes \mathbbm{1}_E) \rho_{XBE}^{\otimes n} (\Pi \otimes \mathbbm{1}_E)$ is also a purification of $\tau_{X^n B^n}$. Consequently, the followings identities hold: $H(X|B)_\rho = -H(X|E)_\rho$ and $H(X^n|B^n)_{\tau} = - H(X^n|E^n)_{\tau}$.

The Asymptotic Equipartition Property (Corollary $6.5$ from \cite{tom12}) applied to $\rho^{\otimes n}$ gives 
\begin{align}
 H_{\mathrm{min}}^\epsilon(X^n|E^n)_{\rho^{\otimes n}} \geq H(X^n|E^n)_{\rho^{\otimes n}} -2(d+1)\sqrt{n} \log_2 \frac{2}{\epsilon^2}.
\end{align}
Define $\lambda$ such that $-\log_2 \lambda = H(X^n|E^n)_{\rho^{\otimes n}} -2(d+1)\sqrt{n} \log_2 \frac{2}{\epsilon^2}$. 
By definition of the smooth min-entropy, there exists an operator $\bar{\rho}_{X^nE^n}$ with the following properties:
\begin{align}
\lambda \cdot (\mathbbm{1}_X \otimes \rho_E)^{\otimes n} - \bar{\rho}_{X^n E^n} &\geq 0. \\
 \|\rho_{XE}^{\otimes n} - \bar{\rho}_{X^n E^n}\|_1 &\leq \epsilon.
\end{align}
Define also the operator $\bar{\tau}_{X^n B^n E^n} :=  \frac{1}{p} \Pi \bar{\rho}_{X^n B^n E^n} \Pi$. 
Since $\Pi \leq \mathbbm{1}_{XBE}$, the following inequalities hold:
\begin{align}
\|\tau_{X^n E^n}-\bar{\tau}_{X^n E^n}\|_1 \leq \frac{\epsilon}{p} \quad \text{and} \quad
\bar{\tau}_{X^n E^n} & \leq \frac{1}{p} \bar{\rho}_{X^n E^n} \leq \frac{1}{p} \lambda  \cdot (\mathrm{id}_X \otimes \rho_E)^{\otimes n}.
\end{align}
By definition of the conditional von Neumann entropy,
\begin{align}
H(X^n|E^n)_{\bar{\tau}_{X^n E^n}} & := \max_{\sigma_E^n} \mathrm{tr}\, \left[ \bar{\tau}_{X^n E^n} \left( \mathbbm{1}_X \otimes \log_2 \sigma_E^{\otimes n} -\log_2 \bar{\tau}_{X^n E^n} \right) \right]\\
&\geq \mathrm{tr} \left[ \bar{\tau}_{X^n E^n} \left( \mathbbm{1}_X \otimes \log_2 \rho_E^{\otimes n} -\log_2 \bar{\tau}_{X^n E^n}\right)\right]\\
&\geq \mathrm{tr} \left[ \bar{\tau}_{X^n E^n} \left( \mathbbm{1}_X \otimes \log_2 \rho_E^{\otimes n}-\log_2 (\mathbbm{1}_X \otimes \rho_E)^{\otimes n} - \log_2 \lambda - \log_2 \frac{1}{p}\right) \right]\\
&\geq \mathrm{tr} \left[ \bar{\tau}_{X^n E^n} \left(- \log_2 \lambda - \log_2 \frac{1}{p}\right) \right]\\
&\geq  \mathrm{tr} \left[ \bar{\tau}_{X^n E^n} \right] \left[ H(X^n|E^n)_{\rho^{\otimes n}} -2(d+1)\sqrt{n} \log_2 \frac{2}{\epsilon^2}- \log_2 \frac{1}{p}\right]\\
&\geq \left(1 -\frac{\epsilon}{p} \right)H(X^n|E^n)_{\rho^{\otimes n}} -4(d+1)\sqrt{n} \log_2 \frac{2}{\epsilon^2}- 2\log_2 \frac{1}{p}.
\end{align}
The duality property of the conditional von Neumann entropy implies that:
\begin{align}
H(X^n|B^n)_{\rho^{\otimes n}} \geq \left(1 -\frac{\epsilon}{p} \right)H(X^n|E^n)_{\rho^{\otimes n}} \geq H(X^n|B^n)_{\bar{\tau}_{X^n B^n}} - 4(d+1)\sqrt{n} \log_2 \frac{2}{\epsilon^2} - 2\log_2 \frac{1}{p}.
\end{align}
The Alicky-Fannes inequality applied to $\tau$ and $\bar{\tau}$ gives:
\begin{align}
H(X^n|B^n)_{\bar{\tau}_{X^n B^n}} \geq H(X|B)_{\tau_{X^n B^n}} - 4 \frac{\epsilon d}{p} - 2 h(\epsilon/p).
\end{align}
Finally, noticing that the binary entropy $h(\epsilon/p)$ is less than 1 yields:
\begin{align}
H(X^n|B^n)_{\rho^{\otimes n}} \geq H(X^n|B^n)_{{\tau}_{X^n B^n}} - 4(d+1)\sqrt{n} \log_2 \frac{2}{\epsilon^2} - 2\log_2 \frac{2}{p} - 4 \frac{\epsilon d}{p}.
\end{align}
Combining this bound with Eq.~\ref{partial-proof} completes the proof.
\end{proof}

If $p \geq \epsilon$, we obtain the bound:
\begin{align}
H_{\mathrm{min}}^{\epsilon_{\mathrm{sm}}} (X^{2n}|E^{2n})_{\tau_{X^{2n} E^{2n}}}  \geq H(X^{2n}|E^{2n})_{{\tau}_{X^{2n} E^{2n}}} - \Delta_{AEP},
\end{align}
with
\begin{align}
\label{def-delta-AEP}
\Delta_{AEP} := \sqrt{2n} \left[(d+1)^2 + 4(d+1) \log_2 \frac{2}{\epsilon_{\mathrm{sm}}^2} + 2\log_2 \frac{2}{\epsilon^2\epsilon_{\mathrm{sm}}}\right] + 4 \frac{\epsilon_{\mathrm{sm}} d}{\epsilon}.
\end{align}

\subsection{Lower bound on the entropy of an i.i.d. variable}

Let $U$ be a random variable with finite support $\{u_1, \ldots, u_{2^d}\}$  described by an unknown probability distribution $p_i =  \mathrm{Pr}[U=u_i]$. Given that we observe $n$ independent realizations of this random variable, $U = (U_1, \ldots, U_n)$, we wish to obtain a lower bound on its entropy $H(U) = - \sum_{i=1}^{2^d} p_i \log p_i$. We define the Maximum Likelihood Estimator (MLE), also know as ``empirical entropy", for $H(U)$ to be 
\begin{align}
\hat{H}_{\mathrm{MLE}}(U):= - \sum_{i=1}^{2^d} \hat{p}_i \log \hat{p}_i,
\end{align}
where $\hat{p}_i = \frac{1}{n} \sum_{k=1}^n \delta_i(U_k)$ and $\delta_i$ denotes the probability measure concentrated at $u_i$.
The random variable $\hat{H}_{\mathrm{MLE}}(U)$ is negatively biased everywhere \cite{pan03}:
\begin{align}
\mathbbm{E}_p\hat{H}_{\mathrm{MLE}}(U) \leq H(U),
\end{align}
where $\mathbbm{E}_p$ denotes the conditional expectation given $p$.
This gives a lower bound on $H(U)$ but one that is not directly observable in an experiment. The missing piece is a concentration result for $\hat{H}_{\mathrm{MLE}}(U)$ around its mean due to Antos and Kontoyiannis \cite{AK01}.
\begin{theo}[Antos and Kontoyiannis]
\begin{align}
\mathrm{Pr}\left[\left|\hat{H}_{\mathrm{MLE}}(U) - \mathbbm{E} \hat{H}_{\mathrm{MLE}} \right|\geq \delta\right] \leq \epsilon,
\end{align}
where 
$$
\delta = \sqrt{\frac{2 \log_2^2 n \log(2/\epsilon)}{n}}.
$$
\end{theo}
Since $\hat{H}_{\mathrm{MLE}}(U)$ is directly measurable in an experiment, we obtain a lower bound on the entropy of the distribution, when we are given access to $n$ i.i.d.~realizations:
\begin{align}
\label{AEP-U}
\mathrm{Pr}\left[H(U) \leq \hat{H}_{\mathrm{MLE}}(U)-  \sqrt{\frac{2 \log_2^2 n \log(2/\epsilon)}{n}}\right] \leq \epsilon.
\end{align}

Finally, we consider  the random variable $V$ with support $\{ u_1, \ldots, u_{2^d}\}^{n}$ and probability distribution 
\begin{align}
\mathrm{Pr}[V = u_{i_1} \ldots u_{i_n}] = \left\{
\begin{array}{ccc}
\frac{p_{i_1} \cdots p_{i_n}}{p_A } & \mathrm{if} & u_{i_1} \ldots u_{i_n} \in A,\\
0 & \mathrm{otherwise}.& 
\end{array}
\right.
\end{align}
where $p_A := \mathrm{Pr}[U^n \in A]$ is the probability that the string $U^n$ belongs to the set $A$. 
It immediately follows from Eq.~\ref{AEP-U} that:
\begin{align}
\label{AEP-V}
\mathrm{Pr}\left[H(V) \leq \hat{H}_{\mathrm{MLE}}(V) - \log_2 \frac{1}{p_A}-  \sqrt{\frac{2 \log_2^2 n \log(2/\epsilon)}{n}}\right] \leq \frac{\epsilon}{p_A}.
\end{align}

In the QKD protocol, the probability $p_A$ of passing can be assumed to be at least equal to the security parameter $\epsilon$. We obtain that the following bound holds, except with probability $\epsilon_{\mathrm{end}}/\epsilon$:
\begin{align}
4n H(U) \geq 4n  \hat{H}_{\mathrm{MLE}}(U) - \Delta_{\mathrm{ent}}
\end{align}
where $U$ is the string of size $4n$ corresponding to the raw key, and where 
\begin{align}
\label{def-delta-ent}
\Delta_{\mathrm{ent}} :=  \log_2 \frac{1}{\epsilon}+  \sqrt{8n \log_2^2 (4 n) \log(2/\epsilon_{\mathrm{sm}})}.
\end{align}

\subsection{Gaussian states and covariance matrices}

Using the AEP, the problem of computing the smooth min-entropy above can actually be reduced to the problem of computing the Holevo information $\chi(Y;E)$ between Bob's measurement outcome $Y$ and the register $E$. 
Thanks to extremality properties of Gaussian states \cite{WGC06, GC06}, it is known that this quantity can be upper bounded by its value computed for a Gaussian state with the same covariance matrix as the true state. 
We will now proceed and show that it is in fact sufficient to know a symmetrized covariance matrix in order to get a bound on Eve's information.

We need to analyze the symmetrization of the covariance matrix and show that without loss of security, one can assume that it only depends on 3 variables. 
A general 2-mode covariance matrix with the appropriate symmetry (that is, where the state has been symmetrized with the map $\mathrm{Sym}$ described in Eq.~\ref{sym-map}) is of the form:
\begin{align}
\gamma= \left[
\begin{matrix}
x & 0 & z \cos \theta & z \sin \theta \\
0 & x & z \sin \theta & -z \cos \theta\\
z \cos \theta & z \sin \theta  & y & 0\\
z \sin \theta & -z \cos \theta& 0 & y\\
\end{matrix}
\right]
\end{align}
One can check that the determinant $D$ of $\gamma$ and the quantity $\Delta = x^2 + y^2 -2z^2$ (corresponding to the sum of the determinants of the four $2\times 2$ blocks of $\gamma$) are independent of $\theta$. This means that the symplectic eigenvalues of $\gamma$ are independent of $\theta$. This was expected since $\theta$ corresponds to a phase-shift applied to Alice's mode for instance. 

The conditional $2 \times 2$ covariance matrix of Alice's state given that Bob performed a heterodyne detection on his part of the state is 
\begin{align}
\gamma_{A|y}^{\mathrm{het}} & = \gamma_A - \gamma_C (\gamma_B + \mathbbm{1}_2)^{-1} \gamma_C^T\\
&= \left[ \begin{matrix}  x-\frac{z^2}{1+y} & 0 \\ 0 & x - \frac{z^2}{1+y} \end{matrix}\right]
\end{align}
which is also independent of $\theta$. Here, $\gamma_A, \gamma_B$ and $\gamma_C$ are the $2\times 2$ blocks of $\gamma$.

We now recall that the Holevo information $\chi(Y;E)$ computed for the Gaussian state of variance $\gamma$ is given by \cite{WPG12}
\begin{align}
\chi(Y;E) = g\left[\frac{\nu_1-1}{2} \right] + g\left[\frac{\nu_2-1}{2} \right] - g\left[\frac{\nu_3-1}{2} \right],
\end{align}
where $\nu_1, \nu_2$ are the symplectic eigenvalues of $\gamma$ and $\nu_3$ is the symplectic eigenvalue of $\gamma_{A|y}^{\mathrm{het}} $ and $g(x) := (x+1) \log(x+1) -x \log (x)$.
Let us introduce the function $f(x,y,z,\theta) := \chi(Y;E)$ for the Gaussian state with covariance matrix $\gamma$.

If $\theta=0$ and $x$ and $y$ are fixed, then one can check numerically that $f$ is monotonically decreasing when $z$ is increasing, i.e. the function $z \mapsto f(x,y,z,0)$ is decreasing. For simplicity, we write $f(x,y,z)$ when the fourth variable is zero. This has an intuitive interpretation: if the correlations between Alice and Bob decrease, then Eve's information increases.
This implies that $f(x,y,z \cos \theta, 0) \geq f(x,y,z,0) = f(x, y, z, \theta)$.
In other words, one can always assume that the covariance matrix $\gamma$ of the symmetrized state has the form 
\begin{align}
\gamma= \left[
\begin{matrix}
x & 0 & z & 0 \\
0 & x & 0 & -z \\
z & 0  & y & 0\\
0 & -z & 0 & y\\
\end{matrix}
\right].
\end{align}
This means that in order to upper bound the Holevo information between Bob's measurement result and Eve for a state $\tilde{\rho}^n$, it is sufficient to obtain bounds on $\Sigma_a, \Sigma_b$ and $\Sigma_c $ defined in Eq.~\ref{cov-matrix-elements}.

We now devote a section to give more details about the parameter estimation procedure.

\section{Parameter Estimation in the protocol $\mathcal{E}_0$}
\label{section-PE}

We first describe the intuition behind the parameter estimation before proving some technical statements. 

\subsection{Principle}

The Parameter Estimation procedure can be decomposed into two steps. First, Alice needs to compute a confidence region for the three quantities $\|X\|^2, \|Y\|^2$ and $\langle X,Y\rangle$. Then, she should be able to obtain bounds on the covariance matrix of the state she shared with Bob.

As we already mentioned, we assume here that $\|X\|^2, \|Y\|^2$ and $\langle X,Y\rangle$ are known. 
We relate the Parameter Estimation procedure to a certain GendankenExperiment, which is described in Fig.~\ref{Gedanken}. This experiment involves two additional players $A_1$ and $A_2$ on Alice's side as well as two others $B_1$ and $B_2$ on Bob's side. Then, we will show that the GedankenExperiment can in fact be efficiently simulated by Alice and Bob alone. 

In the GedankenExperiment, Alice and Bob start by symmetrizing their quantum state by processing it through a random network of beamsplitters and phase shifters. Then, they split their respective $2n$ modes into two sets of $n$ modes, which they forward to $A_1$ and $B_1$ for the first half, and to $A_2$ and $B_2$ for the second half. 
Now, both couples ($A_1, B_1$ and $A_2, B_2$) are in position to perform a conventional quantum state tomography procedure:  the couple of players $A_1, B_1$ can try to estimate the state held by $A_2$ and $B_2$ and vice versa. Players $A_1$ and $B_1$ will be able to infer a lower bound on the secret key rate that $A_2$ and $B_2$ can extract from their state. Similarly, $A_2$ and $B_2$ will be able to infer a lower bound on the secret key rate that $A_1$ and $B_1$ can extract from their state. Adding the two bounds gives a bound on the total secret key rate that Alice and Bob would have been able to extract from their overall states, had they known the results of the two parameter estimation procedures. 
Then, we can show that because the measurement is the same for key elements and for parameter estimation, Alice and Bob are able to simulate the results of the actions of $A_1, A_2, B_1$ and $B_2$.

\begin{figure}[htbp!]
\begin{framed}
  \centering
\begin{enumerate}
\item \textbf{State Preparation:} Alice and Bob each have access to the global state $\rho_{AB}^{2n}$.
\item \textbf{State Symmetrization:} A random unitary $V$ is drawn from the Haar measure on $U(2n)$. The state $\rho_{AB}^{2n}$ is mapped to $\left(\Phi(V)_A \otimes \Phi(V^*)_B\right) \rho_{AB}^{2n} \left(\Phi(V)_A \otimes \Phi(V^*)_B\right)^\dagger$. This can be achieved by processing the optical modes through the appropriate network of beamsplitters and phase-shifts. 
\item \textbf{Distribution to additional players:} Alice forwards her first $n$ modes to Agent $A_1$ and the remaining $n$ modes to Agent $A_2$. Similarly, Bob forwards his first $n$ modes to $B_1$ and the remaining $n$ modes to $B_2$. The first $n$ modes correspond to a bipartite state $\rho_1$, the remaining modes to a state $\rho_2$.
\item \textbf{Measurement:} Agents $A_1$ and $B_1$ measure their $n$ respective modes with a heterodyne detection obtaining two vectors $X_1$ and $Y_1$ of length $2n$. 
Similarly, $A_2$ and $B_2$ measure their modes with heterodyne detection, obtain vectors $X_2$ and $Y_2$. 
Finally, agents $B_1$ and $B_2$ publicly reveal $Y_1$ and $Y_2$.
\item \textbf{Parameter Estimation:}
Agent $A_1$ uses $X_1$ and $Y_1$ to compute a confidence region for the (averaged) covariance matrix of $\rho_2$. Similarly, $A_2$ computes a confidence region for that of $\rho_1$. 
\end{enumerate} 

\end{framed}
\caption{Parameter Estimation Procedure (GedankenExperiment)}
\label{Gedanken}
\end{figure}

In order to analyze the GedankenExperiment, we show that it can be simulated by Alice if she initially knows the values of $\|X\|^2, \|Y\|^2$ and $\langle X,Y\rangle$. There are two main steps:
\begin{itemize}
\item Alice needs first to simulate the symmetrization and the distribution of the first $n$ modes to $A_1$. In particular, she needs to compute a confidence region for the three parameters $\|X_1\|^2, \|Y_1\|^2$ and $\langle X_1, Y_1\rangle$, which correspond to the norms and dot product of the vectors of measurements outcomes for $A_1$ and $B_1$. This will be analyzed in Lemmas \ref{estX} and \ref{estXY}.
\item Second, Alice needs to compute the confidence regions for the state $\tilde{\rho}_2^n$ that $A_1$ and $B_1$ would have inferred from their measurement results $X_1$ and $Y_1$. This is the object of Lemma \ref{X1givesX2}.
\end{itemize}

\subsection{Proofs related to the analysis of Parameter Estimation}

The first lemma is a standard concentration measure result in statistics. 
\begin{lemma}[Tail bounds for $\chi^2$ distribution \cite{LM00}]
\label{laurent-massart}
Let $U$ be a $\chi^2$ statistics with $n$ degrees of freedom. For any $x>0$, 
\begin{align}
\mathrm{Pr}\left[U -n \geq 2\sqrt{nx} + 2x \right] \leq e^{-x} \quad \mathrm{and} \quad
\mathrm{Pr}\left[U -X \geq 2\sqrt{nx}  \right] \leq e^{-x}.
\end{align}
\end{lemma}

The following two lemmas deal with the situation where Alice tries to simulate the distribution of $\rho_1$ to $A_1$ and $B_2$ and $\rho_2$ to $A_2$ and $B_2$. In particular, given the knowledge of $\|X\|^2, \|Y\|^2$ and $\langle X,Y\rangle$, she can compute confidence regions for $\|X_i\|^2, \|Y_i\|^2$ and $\langle X_i,Y_i\rangle$.

Since the symmetrization of the state with the map $\mathrm{Sym}$ described \ref{sym-map} commutes with the heterodyne measurement, Alice can simulate the measurement of $A_1$ by first measuring $X$ and only later perform the random rotation. In particular, the random vector $X_1$ corresponds to the projection of $X$ on a random subspace of complex dimension $n$. 
\begin{lemma}
\label{estX}
Given a vector $X \in \mathbbm{C}^{2n}$, consider $X_1$ the projection of $X$ on a random subspace of dimension $n$, then for $\epsilon \geq 2e^{-n/2}$,
\begin{align}
\mathrm{Pr}\left[2\|X_1\|^2 \geq \left[1+1.5\sqrt{\frac{\ln (2/\epsilon)}{n}}\right] \|X\|^2 \right]& \leq\epsilon \\
\mathrm{Pr}\left[2\|X_1\|^2 \leq\left[1-2.2\sqrt{\frac{\ln (2\epsilon)}{n}}\right] \|X\|^2 \right] &\leq \epsilon.
\end{align}
\end{lemma}

\begin{proof}
By rotation invariance of the problem in $\mathbbm{C}^{2n}$, one can fix the random subspace to be $\mathrm{Span}(|1\rangle, \ldots, |n\rangle)$ where $(|1\rangle, \ldots, |2n\rangle)$ is the canonical basis of $\mathbbm{C}^{2n}$ and take the vector $|X\rangle$ to be uniformly distributed on the sphere of radius $\|X\|$ in $\mathbbm{C}^{2n}$. We wish to get bounds on the random variable $\langle X |\Pi |X\rangle$ where $\Pi = \sum_{i=1}^n |i\rangle \langle i|$.
Writing $|X\rangle = \sum_{i=1}^{2n} \alpha_i |i\rangle$, one obtains 
\begin{align}
\langle X |\Pi |X\rangle = \sum_{i=1}^n |\alpha_i|^2 = \frac{\sum_{i=1}^n |\alpha_i|^2}{\sum_{i=1}^n |\alpha_i|^2 + \sum_{i=n+1}^{2n} |\alpha_i|^2} = \frac{\sum_{i=1}^n x_i^2}{\sum_{i=1}^n x_i^2 + \sum_{i=1}^{n} y_i^2},
\end{align}
where $x_i$ and $y_i$ are i.i.d.~normal variables $\mathcal{N}(0,1)$ and the last equality follows from the well-known fact that one can generate  a uniformly distributed vector on a sphere by drawing i.i.d.~normal variables for its coordinates and normalizing the resulting vector. 
In particular, $\langle X |\Pi |X\rangle $ has the same distribution as $\frac{U}{U + U'}$ where $U$ and $U'$ are independent $\chi^2$ random variables.

Using the bounds of Lemma \ref{laurent-massart} together with the union bound, one obtains
\begin{align*}
\mathrm{Pr}\left[U  \geq n + 2\sqrt{nx} + 2x \right] \leq e^{-x} & \quad \text{and} \quad \mathrm{Pr}\left[U  \leq n -2\sqrt{nx}  \right]\leq e^{-x},\\
\mathrm{Pr}\left[\frac{U}{U+U'}  \geq \frac{n+2\sqrt{nx} + 2x}{ 2(n+x ) } \right] \leq 2e^{-x} & \quad\text{and} \quad\mathrm{Pr}\left[\frac{U}{U+U'}  \leq \frac{n-2\sqrt{nx} }{2(n +x)} \right] \leq 2e^{-x} \\
\end{align*}

Fix $\gamma = \sqrt{x/n}$. One can easily check that for $\gamma \in [0, 1]$,
\begin{align}
\frac{1+2\gamma + 2\gamma^2}{2(1+\gamma)}  \leq \frac{1}{2}\left[1 + \frac{3\gamma}{2}\right] \quad \text{and} \quad
\frac{1-2\gamma}{2(1+\gamma^2)} \geq \frac{1}{2}(1-2.2\gamma),
\end{align}
concluding the proof.
\end{proof}

\begin{lemma}
\label{estXY}
Given two vectors $X, Y \in \mathbbm{R}^{4n}$, consider $X_1$ and $Y_1$ the projections of $X$ and $Y$ on a random subspace of dimension $2n$, and $x \leq n/2$, then 
\begin{align}
\mathrm{Pr}\left[\langle X,Y\rangle - 1.85\sqrt{\frac{x}{n}} \left( \|X\|^2 +\|Y\|^2 \right) \leq 2\langle X_i, Y_i\rangle \leq \langle X,Y\rangle + 1.85\sqrt{\frac{x}{n}} \left( \|X\|^2 +\|Y\|^2 \right) \right]  \geq 1- 8 e^{-x}.
\end{align}
\end{lemma}

\begin{proof}
The proof can be reduced to the result of Lemma \ref{estX} since 
\begin{align}
\langle X, Y\rangle = \frac{1}{4} \left[\|X+Y\|^2 - \|X-Y\|^2 \right] \quad \text{and} \quad \langle X_i, Y_i\rangle = \frac{1}{4} \left[\|X_i+Y_i\|^2 - \|X_i-Y_i\|^2 \right].
\end{align}
The union bound insures that, except with probability at most $8e^{-x}$, one has (for $x \leq n$),
\begin{align}
\left[1-2.2 \sqrt{\frac{x}{n}} \right] \|X+Y\|^2 &\leq 2 \|X_1 + Y_1 \|^2 \leq \left[1+1.5 \sqrt{\frac{x}{n}} \right] \|X+Y\|^2\\
-\left[1+1.5 \sqrt{\frac{x}{n}} \right] \|X-Y\|^2 &\leq - 2 \|X_1 - Y_1 \|^2 \leq -\left[1-2.2 \sqrt{\frac{x}{n}} \right] \|X-Y\|^2.
\end{align}
Summing these inequalities and using that $2\langle X, Y\rangle \leq \|X\|^2 + \|Y\|^2$ gives:
\begin{align}
\langle X,Y\rangle - 1.85\sqrt{\frac{x}{n}} \left( \|X\|^2 +\|Y\|^2 \right) \leq 2\langle X_i, Y_i\rangle \leq \langle X,Y\rangle + 1.85\sqrt{\frac{x}{n}} \left( \|X\|^2 +\|Y\|^2 \right).
\end{align}
\end{proof}
Note in particular that the following bounds hold:
\begin{align}
\mathrm{Pr}\left[ \langle X_i, Y_i \rangle \leq \frac{1}{2} \langle X,Y\rangle - \sqrt{\frac{x}{n}} ( \|X\|^2 + \|Y\|^2) \right] & \leq 4 e^{-x} \label{bound-c}\\
\mathrm{Pr}\left[| \langle X_1, Y_1 \rangle - \langle X_2,Y_2\rangle| \geq 2 \sqrt{\frac{x}{n}} ( \|X\|^2 + \|Y\|^2) \right] & \leq 8 e^{-x}.
\end{align}

We now reason in terms of the covariance matrix of the Husimi $Q$-function, that is the function giving the probability density function for the result of a heterodyne measurement. 
The averaged covariance matrix of the $Q$-function of a (randomized) state $\rho_{AB}^{n}$ is 
\begin{align}
\Gamma_Q = \bigoplus_{k=1}^n \left[ 
\begin{matrix}
\frac{1}{2}(\Sigma_a +1) & 0 & \frac{\Sigma_c}{2} & * \\
0 & \frac{1}{2}(\Sigma_a +1) & * & -\frac{\Sigma_c}{2} \\
\frac{\Sigma_c}{2} & * & \frac{1}{2}(\Sigma_b +1) & 0 \\
* & -\frac{\Sigma_c}{2} & 0 & \frac{1}{2}(\Sigma_b +1) \\
\end{matrix}\right],
\end{align}
where some entries are not specified ($*$).
Because Alice has access to the heterodyne measurement of the first half of the (symmetrized) state $\tilde{\rho}^{2n}$, she is able to infer bounds on measurement outcomes for the second half, that will hold except with some small probability. 

The following lemma deals with the scenario where players $A_1$ and $B_1$ try to estimate the covariance matrix of $\tilde{\rho}^n_2$, given their measurement outcomes of $\tilde{\rho}_1^n$, namely the quantities $\|X_1\|^2, \|Y_1\|^2$ and $\langle X_1, Y_1\rangle$.

\begin{lemma} Consider an $8n$-dimensional probability distribution $Q(X, Y)$ where $X = (X_1,X_2) \in \mathbbm{R}^{4n}$ and $Y=(Y_1,Y_2) \in \mathbbm{R}^{4n}$, which is rotationally-invariant (when applying an orthogonal transformation and its transpose to $X$ and $Y$). 
Then, for ${\frac{\log(2/\epsilon)}{2n}} \leq 0.05$, the following bounds hold:
\begin{align}
&\mathrm{Pr}\left[ \|X_2\|^2 \geq \left[1+ 5\sqrt{\frac{\log(2/\epsilon)}{2n}}\right] \|X_1\|^2\right] \leq \epsilon, \label{ineqA}\\
&\mathrm{Pr}\left[ \|Y_2\|^2 \geq \left[1+ 5\sqrt{\frac{\log(2/\epsilon)}{2n}}\right]  \|Y_1\|^2\right] \leq \epsilon,\label{ineqB}\\
&\mathrm{Pr}\left[ \langle X_2, Y_2\rangle \leq \langle X_1, Y_1\rangle- \frac{9}{2} \sqrt{\frac{\log(2/\epsilon)}{2n}} \left(\|X_1\|^2+\|Y_1\|^2\right) \right] \leq 2\epsilon. \label{ineqC}
\end{align}
\label{X1givesX2}
\end{lemma}

\begin{proof}
The first two inequalities are a direct application of Lemma $B.1$ in Ref.~\cite{LGR13} and of the observation that for $\gamma \in [0, 0.05]$, the following inequality holds:
\begin{align}
\frac{1+2\gamma+2\gamma^2}{1-2\gamma} \geq 1+5\gamma.
\end{align}
Similarly, one can show that the following inequalities also hold:
\begin{align}
\mathrm{Pr}\left[ \|X_2\|^2 \leq \left[1-4\sqrt{\frac{\log(2/\epsilon)}{2n}}\right] \|X_1\|^2\right] \leq \epsilon, \label{ineqD}\\
\mathrm{Pr}\left[ \|Y_2\|^2 \leq \left[1-4\sqrt{\frac{\log(2/\epsilon)}{2n}}\right]  \|Y_1\|^2\right] \leq \epsilon. \label{ineqE}
\end{align}
Using the same strategy as in the proof of Lemma \ref{estXY}, one can apply Inequalities \ref{ineqA}, \ref{ineqB}, \ref{ineqD} and \ref{ineqE} to vectors $X_i \pm Y_i$, which  immediately gives Inequality \ref{ineqC} thanks to the union bound. 
\end{proof}

We will also need techniques to bound the expectation of some random variables, and not only the values of the random variables. 

Let $p(x,y)$ be a probability distribution defined on $[0, \infty[^2$ such that 
\begin{align}
\int_{0}^a \int_b^{\infty} p(x,y) \mathrm{d}x \mathrm{d}y \leq \epsilon(a,b). 
\end{align}
We define two sets $A$ and $B_\delta$ for the random variable $X$:
\begin{align}
A & := \left\{x \in [0 ,a] \right\} \quad \text{and} \quad
B_\delta  := \left\{ x \: : \: \int_{0}^\infty y p(y|x) \mathrm{d}y \geq b + \delta \right\}.
\end{align}

\begin{lemma}
\label{bound-var}
Let $(b_k)_{k=1 ..\infty}$ a nondecreasing sequence such that $b_0=b$ and $\lim_{k\rightarrow \infty} = \infty$, then
 \begin{align}
 \mathbbm{P}[ A \cap B_\delta] \leq \frac{1}{\delta} \sum_{k=1}^\infty b_{k+1} \epsilon(a, b_k).
\end{align}
\end{lemma}

\begin{proof}
We wish to compute the probability that $x \in A \cap  B_\delta$. 
Let us first compute the expectation of $Y$ over that set:
\begin{align}
\int_{x \in A \cap B_{\delta} } \mathrm{d}x p(x) \int_0^\infty y p(y|x) \mathrm{d}y & = \int_{x \in A\cap B_{\delta}} \mathrm{d}x p(x) \left[ \int_0^b y p(y|x) \mathrm{d}y + \sum_{k=1}^\infty \int_{b_k}^{b_{k+1}} y p(y|x) \mathrm{d}y  \right] \\
&\leq \int_{x \in A\cap B_{\delta}} b  p(x) \mathrm{d}x   + \sum_{k=1}^\infty b_{k+1} \epsilon(a, b_k) \\
&= b \mathbbm{P}[x \in A\cap B_{\delta} ]       + \sum_{k=1}^\infty b_{k+1} \epsilon(a, b_k).
\end{align}
On the other hand, we know that:
\begin{align}
\int_{x \in A \cap B_\delta} \mathrm{d}x p(x) \int_0^\infty y p(y|x) \mathrm{d}y  \geq (b+\delta)  \mathbbm{P}[x \in A \cap B_\delta].
\end{align}
Putting both inequalities together yields:
\begin{align}
(b+\delta)  \mathbbm{P}[x \in A \cap B_\delta] \leq b \mathbbm{P}[x \in A\cap B_{\delta} ]       + \sum_{k=1}^\infty b_{k+1} \epsilon(a, b_k),
\end{align} 
which completes the proof.
\end{proof}

Let us define two sets $C$ and $D_\delta$ for the random variable $X$:
\begin{align}
C  := \left\{x \in [a, \infty[ \right\} \quad \text{and} \quad 
D_\delta  := \left\{ x \: : \: \int_{0}^\infty y p(y|x) \mathrm{d}y \leq b - \delta \right\}.
\end{align}
We have the following lemma. 
\begin{lemma}
\label{lemma-bound-cov}
If $p(x,y)$ is such that $ \int_{a}^\infty \int_0^{b} p(x,y) \mathrm{d}x \mathrm{d}y \leq \epsilon_2(a,b)$,
then
\begin{align}
 \mathbbm{P}[ C \cap D_\delta] \leq \frac{b \epsilon_2(a,b)}{\delta}.
\end{align}
\end{lemma}
\begin{proof}
We wish to compute the probability that $x \in C \cap  D_\delta$. 
Let us first compute the expectation of $Y$ over that set:
\begin{align}
\int_{x \in C \cap D_{\delta} } \mathrm{d}x p(x) \int_0^\infty y p(y|x) \mathrm{d}y & = \int_{x \in C\cap D_{\delta}} \mathrm{d}x p(x) \left[ \int_0^{b} y p(y|x) \mathrm{d}y + \int_{b}^{\infty} y p(y|x) \mathrm{d}y  \right] \\
& = b \mathbbm{P}[x \in C\cap D_{\delta} ] + \int_{x \in C\cap D_{\delta}} \mathrm{d}x p(x) \left[ \int_0^{b} (y-b) p(y|x) \mathrm{d}y+ \int_{b}^{\infty} (y-b) p(y|x) \mathrm{d}y  \right] \\
&\geq b \mathbbm{P}[x \in C\cap D_{\delta} ] - b \epsilon_2(a,b)
\end{align}
On the other hand, we know that:
\begin{align}
\int_{x \in C \cap D_\delta} \mathrm{d}x p(x) \int_0^\infty y p(y|x) \mathrm{d}y  \leq (b-\delta)  \mathbbm{P}[x \in C \cap D_\delta].
\end{align}
Putting both inequalities together yields:
\begin{align}
b \mathbbm{P}[x \in C\cap D_{\delta} ] - b \epsilon_2(a,b) \leq(b-\delta)  \mathbbm{P}[x \in C \cap D_\delta],
\end{align} 
which concludes the proof. 
\end{proof}

We are now in a position to prove our main result regarding parameter estimation. 

\subsection{Probability of the bad event}


Consider that the PE test passed, i.e. that $[\gamma_a \leq \Sigma_a^{\max} ] \wedge [\gamma_b \leq \Sigma_b^{\max} ] \wedge [\gamma_c \geq \Sigma_c^{\min} ]$. The problematic cases are the ones where either $A_1$ or $A_2$ obtains values not compatible with their own parameter estimation procedure (i.e. they would abort their own protocol), or when these procedures fail (i.e. the protocol did not abort but the estimation is incorrect). 
For instance, the problematic cases for the estimation of $A_1$ and $A_2$ variances correspond to 
\begin{align}
E_{\mathrm{bad}}^{\|X\|^2} := [\|X_1\|^2 \geq a] \vee [\|X_2\|^2 \geq a] \vee [ (\mathbbm{E} \|X_2\|^2 \geq b) \wedge (\|X_1\|^2 \leq a)]  \vee [ (\mathbbm{E} \|X_1\|^2 \geq b) \wedge( \|X_2\|^2 \leq a)],
\end{align}
where $a$ and $b$ will be optimized later.
Similarly, one can define:
\begin{align}
E_{\mathrm{bad}}^{\|Y\|^2} := [\|Y_1\|^2 \geq a] \vee [\|Y_2\|^2 \geq a] \vee [ (\mathbbm{E} \|Y_2\|^2 \geq b) \wedge (\|Y_1\|^2 \leq a)]  \vee [ (\mathbbm{E} \|Y_1\|^2 \geq b) \wedge( \|Y_2\|^2 \leq a)]
\end{align}
which is the bad event for the estimation of $B_1$ and $B_2$'s variances.  
Finally, the bad event for the correlations is:
\begin{align}
E_{\mathrm{bad}}^{\langle X,Y\rangle} := [\langle X_1, Y_1\rangle \leq c] \vee [\langle X_2, Y_2\rangle \leq c] \vee [ (\mathbbm{E} \langle X_2, Y_2\rangle \leq d) \wedge (\langle X_1, Y_1\rangle \geq c)]  \vee [ (\mathbbm{E} \langle X_1, Y_1\rangle \leq d) \wedge( \langle X_2, Y_2\rangle \geq c)].
\end{align}

\begin{theo}
\label{main-parameter-estimation}
The probability of the bad event $E_{\mathrm{bad}}^{\\|X\|^2} \vee E_{\mathrm{bad}}^{\|Y\|^2} \vee E_{\mathrm{bad}}^{\langle X,Y\rangle}$ is upper bounded by $\epsilon$ for the following choice of parameters:
\begin{align}
a &  = \frac{1}{2} \left[1+1.5\sqrt{\frac{\log (36/\epsilon)}{n}}\right] \|X\|^2  \quad \text{or} \quad a = \frac{1}{2} \left[1+1.5\sqrt{\frac{\log (36/\epsilon)}{n}}\right] \|Y\|^2,\\
b & = a \left[ 1 + \frac{360}{\epsilon} \exp[-n/25]\right],\\
c &= \frac{1}{2} \langle X, Y\rangle - \sqrt{\frac{\log(72/\epsilon)}{n}} ( \|X\|^2 + \|Y\|^2),\\
d & = c - 2(\|X\|^2 + \|Y\|^2) \sqrt{\frac{8 \log(18/\epsilon)}{n}}.
\end{align}
\end{theo}

\begin{proof}
Using Lemma \ref{estX}, one can define $\epsilon_1(a,y) = 2 \exp\left(-2n\left[ \frac{y-a}{5a}\right]^2 \right)$ such that
\begin{align}
\mathrm{Pr} \left[ (\|X_1\|^2 \leq a) \wedge ( \|X_2\|^2  \geq y )\right] &\leq \mathrm{Pr} \left[ \|X_2\|^2 \geq \frac{y}{a} \|X_1\|^2 \right]\\
& \leq \epsilon_1(a,y). 
\end{align}
Let us define $b_k = b+ k a$ such that $b_0 = b \geq a$ and $\lim_{k \rightarrow \infty} b_k = \infty$. Let us introduce $\gamma = b/a$.
One has:
\begin{align}
\sum_{k=1}^\infty b_{k+1} \epsilon_1(a, b_k) & = 2a \sum_{k=0}^\infty (\gamma  +2 + k ) \exp\left[ -\frac{n}{25} (\gamma + k)^2  \right]\\
& \leq 2a \exp\left[ -n \gamma^2/25 \right] \sum_{k=0}^\infty (\gamma  +2 + k ) \exp\left[ -\frac{2n\gamma k}{25}  \right]\\
& \leq 8a(\gamma+2) \exp\left[ -n \gamma^2/25 \right]
\end{align}
where the last inequality holds provided that $\exp\left[ -\frac{2n\gamma }{25}  \right] \leq 1/2$.
Using Lemma \ref{bound-var}, one obtains that 
\begin{align}
\label{choice-delta}
\mathrm{Pr} \left[ (\|X_1\|^2 \leq a) \wedge ( \mathbbm{E}\|X_2\|^2 \geq a + \delta )\right] \leq \frac{20a}{\delta} \exp\left[ -n/25 \right].
\end{align}
We exploit Lemma \ref{lemma-bound-cov} which states that, for $c \geq d$, 
\begin{align}
\mathrm{Pr}\left[ (\mathbbm{E} \langle X_2, Y_2\rangle \leq d - \delta) \wedge (\langle X_1, Y_1\rangle \geq c) \right] &\leq \frac{d}{\delta} \cdot \mathrm{Pr}\left[ (\langle X_2, Y_2\rangle \leq d) \wedge (\langle X_1, Y_1\rangle \geq c) \right]\\
& \leq \frac{d}{\delta} \cdot \mathrm{Pr}\left[ \langle X_1, Y_1\rangle  - \langle X_2, Y_2\rangle\geq c-d \right]\\
& \leq \frac{8d}{\delta}\exp\left[- n \left[\frac{c-d}{2( \|X\|^2 + \|Y\|^2)}\right]^2\right].
\end{align}
Let us choose $d = c-2 \delta$. Then,
\begin{align}
\mathrm{Pr}\left[ (\mathbbm{E} \langle X_2, Y_2\rangle \leq c -2 \delta) \wedge (\langle X_1, Y_1\rangle \geq c) \right] & \leq \frac{8c}{\delta}\exp\left[- n \left[\frac{\delta}{2( \|X\|^2 + \|Y\|^2)}\right]^2\right]. 
\label{choice-delta2}
\end{align}
Let us choose the values of $a$, $b$, $c$ and $d$ such that each of the 18 individual events has a probability $\epsilon/18$. 
The probability $p_{\mathrm{bad}}^{\mathrm{PE}}$ for the parameter estimation is then:
\begin{align}
p_{\mathrm{bad}}^{\mathrm{PE}} \leq \mathrm{Pr} \left[ E_{\mathrm{bad}}^{\|X\|^2}\right] +\mathrm{Pr} \left[ E_{\mathrm{bad}}^{\|Y\|^2}\right]  + \mathrm{Pr} \left[ E_{\mathrm{bad}}^{\langle X,Y\rangle} \right]   \leq \epsilon.
\end{align}
This is achieved for:
\begin{align}
a &  = \frac{1}{2} \left[1+1.5\sqrt{\frac{\log (36/\epsilon)}{n}}\right] \|X\|^2\\
b & = a \left[ 1 + \frac{360}{\epsilon} \exp[-n/25]\right]\\
c &= \frac{1}{2} \langle X, Y\rangle - \sqrt{\frac{\log(72/\epsilon)}{n}} ( \|X\|^2 + \|Y\|^2)\\
d & = c - 2(\|X\|^2 + \|Y\|^2) \sqrt{\frac{8 \log(18/\epsilon)}{n}}
\end{align}
where the second equality is a consequence of Eq.~\ref{choice-delta}, the third equality is a consequence of Eq.~\ref{bound-c} and the last equality results from Eq.~\ref{choice-delta2} (and noting that $\delta \geq 4c\epsilon/9$ for reasonable parameters). 
\end{proof}

Let us finally define:
\begin{align}
\gamma_a &:= \frac{1}{2n} \left[ 1 + 2\sqrt{\frac{\log(36/\epsilon_{\mathrm{PE}})}{n}}\right] \|X\|^2-1,\\
\gamma_b &:= \frac{1}{2n} \left[ 1 + 2\sqrt{\frac{\log(36/\epsilon_{\mathrm{PE}})}{n}}\right] \|Y\|^2-1,\\
\gamma_c &:= \frac{1}{2n} \langle X, Y\rangle - 5 \sqrt{\frac{\log (8/\epsilon_{\mathrm{PE}})}{n^3}}(\|X\|^2 + \|Y\|^2).
\end{align}
where we choose a regime of $n$ such that 
\begin{align}
 \left[1+1.5\sqrt{\frac{\log (36/\epsilon)}{n}}\right]  \left[1+\frac{360}{\epsilon} \exp[-n/25]\right] \leq  1+2\sqrt{\frac{\log (36/\epsilon)}{n}},
\end{align}
which is the case in all practical situations. 
Moreover, we note that 
\begin{align}
\sqrt{32 \log(18/\epsilon)} + \sqrt{\log(72/\epsilon)}&\leq \sqrt{2(32  \log(18/\epsilon) + \log(72/\epsilon))}\\
& \leq 10 \sqrt{\log(8/\epsilon)} 
\end{align}
We have the following corollary. 
\begin{cor}
The probability that the Parameter Estimation Test passes,  that is, $[\gamma_a \leq \Sigma_a^{\max} ] \wedge [\gamma_b \leq \Sigma_b^{\max} ] \wedge [\gamma_c \geq \Sigma_c^{\min} ] $ and that Eve's information is larger than the one computed for the Gaussian state with covariance matrix characterized by $\Sigma_a^{\max}, \Sigma_b^{\max}$ and $\Sigma_c^{\min}$ is upper-bounded by $\epsilon_{\mathrm{PE}}$. 
\end{cor}

\subsection{Analysis of the Parameter Estimation}

Consider the simulated protocol where Alice and Bob first symmetrize their state, then distribute the first half of their modes to players $A_1$ and $B_2$ and the second half of their modes to $A_2$ and $B_2$. 
The goal of $A_1$ and $B_1$ is to measure their own modes and infer some confidence region for the covariance matrix of the state shared by $A_2$ and $B_2$. Similarly, $A_2$ and $B_2$ measure their state and try to infer a confidence region for the covariance matrix of $A_1$ and $B_1$. 

The bad event that we considered in the previous section corresponds to the case where the parameter estimation test of the true protocol passes but a problem occurs for either one of the virtual parameter estimation tests, i.e. one does not pass (in which case the virtual protocol would have aborted) or both virtual tests pass but their conclusion is not valid (i.e. the covariance matrix of the remaining modes is not in the predicted confidence region). 

Recall that if the state $\tau_{12}$ represents a quantum state on $2n$ modes, and if $\tau_1$ (resp. $\tau_2$) corresponds to the first $n$ (resp. last $n$) modes, then the strong subadditivity of the Holevo information implies that $\chi(Y_1 Y_2;E)_{\tau_{12}} \leq \chi(Y_1;E)_{\tau_1} + \chi(Y_2;E)_{\tau_2}$.

In particular, if we denote by $\tau^{\mathrm{PE}} = \frac{1}{p^{\mathrm{PE}}} \Pi \rho^{\otimes (2n)} \Pi$ (with $p^{\mathrm{PE}} = \mathrm{tr}\,(\Pi \rho^{\otimes (2n)})$) the quantum state conditioned on passing the Parameter Estimation test, then with probability $1-\epsilon/p$, the Holevo information between the string $Y$ corresponding to Bob's heterodyne measurement results and Eve's quantum register is:
\begin{align}
\chi(Y;E)_{\tau^{\mathrm{PE}}} &\leq 2n f(\Sigma_a^{\max}, \Sigma_b^{\max}, \Sigma_c^{\min}), \quad \text{except with probability} \, \epsilon_{\mathrm{PE}}/p.
\end{align}

\section{Security of the protocol $\mathcal{E}_0$ against collective attacks}
\label{sec-proof}

In this section, we finally put the various pieces of the proof together and show that the protocol $\mathcal{E}_0$ described above is secure against collective attacks. 
We will use superscripts $m$, $n$ or $2n$ to recall the length of the various strings.
Our goal is to obtain a lower bound on the smooth min-entropy of the string $U$ given Eve's information, that is, her quantum system $E$ and the public transcript $C$ of the protocol, when the protocol did not abort. The smoothing parameter $\epsilon_{\mathrm{sm}}$ will be optimized later. 

Let us write $\tau = \frac{1}{p} P \rho^{\otimes (2n)} P$ with $p := \mathrm{tr}\, \left( P \rho^{\otimes (2n)}\right)$, the quantum state conditioned on both the Parameter Estimation and the Error Correction tests passing.  Without loss of generality, this passing probability can be assumed to be at least $\epsilon$ since if the abort probability is greater than $1-\epsilon$, the protocol is automatically $\epsilon$-secure. 
Let us define some error parameter $\epsilon_{\mathrm{err}} := \epsilon_{\mathrm{PE}} + \epsilon_{\mathrm{cor}}$.
The analysis of the Parameter Estimation test implies that, except with probability at most $\epsilon_{\mathrm{err}}/p$, one can upper bound the quantum mutual information between $Y$ and $E$ as follows:
\begin{align}
\chi(Y;E)_{\tau} \leq 2n f(\Sigma_a^{\max}, \Sigma_b^{\max}, \Sigma_c^{\min}). 
\end{align}
This is because for a given covariance matrix, the Holevo information is maximized for the Gaussian state with same second moment \cite{GC06} and the function $f$ exactly computes the Holevo information for the Gaussian state with covariance matrix $\bigoplus_{i=1}^{2n} \left[\begin{smallmatrix}
\Sigma_a^{\mathrm{max}} \mathbbm{1}_2 & \Sigma_c^{\mathrm{min}} \sigma_z \\
\Sigma_c^{\mathrm{min}}  \sigma_z & \Sigma_b^{\mathrm{max}}\mathbbm{1}_2 \\
\end{smallmatrix}\right]$.

The chain rule for the smooth min-entropy gives:
\begin{align}
H_{\min}^{\epsilon_{\mathrm{sm}}}( U^{m} | E C)_{\tau} \geq H_{\min}^{\epsilon_{\mathrm{sm}}}( U^{m} | E )_{\tau} - \log |C|.
\end{align}
Recall that $\log |C| = \mathrm{leak}_{\mathrm{EC}}$ since we neglect here the communication where Bob reveals the values of $\|Y\|$ and $\|Y-\hat{Y}\|$ to Alice (which only consumes a small constant number of bits). 
The smooth min-entropy can be bounded further by:
\begin{align}
H_{\min}^{\epsilon_{\mathrm{sm}}}( U^{m} | E)_{\tau} &\geq H(U^{m} |E)_{\tau}  - \Delta_{\mathrm{AEP} } \label{eq-AEP}\\
&= H(U^{m})_{\tau} - \chi(U^{m}; E)_{\tau}  -  \Delta_{\mathrm{AEP} } \label{def-holevo}\\
&\geq4n \hat{H}_{\mathrm{MLE}}(U) -  \Delta_{\mathrm{ent}}   - \chi(Y^{2n}; E)_{\tau}  -  \Delta_{\mathrm{AEP} }\label{data-proc},
\end{align}
where the last inequality holds except with probability $\epsilon_{\mathrm{ent}}/\epsilon$.
Eq.~\ref{eq-AEP} results from the Asymptotic Equipartition Property and $\Delta_{\mathrm{AEP} } $ is defined in Eq.~\ref{def-delta-AEP}.
Eq.~\ref{def-holevo} results from the definition of the Holevo information between the classical string $U^m$ and Eve's quantum system.
Eq.~\ref{data-proc} results from the bound on the entropy as a function of the empirical entropy and from the data-processing inequality and $\Delta_{\mathrm{ent}}$ is defined in Eq.~\ref{def-delta-ent}.

Finally, except with probability $(\epsilon_{\mathrm{err}}+\epsilon_{\mathrm{ent}})/\epsilon$, the smooth min-entropy is lower-bounded as follows:
\begin{align}
H_{\min}^{\epsilon_{\mathrm{sm}}}( U^{m} | E C)_{\tau} \geq 4n \hat{H}_{\mathrm{MLE}}(U) -  2n f(\Sigma_a^{\max}, \Sigma_b^{\max}, \Sigma_c^{\min}) -\mathrm{leak}_{\mathrm{EC}} - \Delta_{\mathrm{ent}} -   \Delta_{\mathrm{AEP} }.
\end{align}

The Leftover Hash Lemma concludes the proof of Theorem \ref{key-rate-theorem-2}.

\section{A security proof against general attacks without active symmetrization}
\label{general-attacks}

We sketch here how to use the Postselection technique \cite{LGR13} to obtain a security proof against general attacks for the case of a direct reconciliation without applying any active symmetrization. The idea is very similar to the Parameter Estimation step of $\mathcal{E}_0$. Alice prepares $4n$ two-mode squeezed vacuum states and sends the appropriate modes to Bob. 
Then Alice and Bob simulate the symmetrization of their state, split their respective modes into two sets of size $2n$ and give the corresponding modes to additional players $A_1$ and $A_2$ for Alice, and $B_1$ and $B_2$ for Bob. Players $A_1$ and $B_1$ perform the energy test, which conditioned on passing, guarantees that the protocol $\mathcal{E}_0$ applied by $A_2$ and $B_2$ will be secure, and similarly players $A_2$ and $B_2$ perform an energy test which gives some security guarantees for the keys obtained by $A_1$ and $B_1$ when applying $\mathcal{E}_0$ to their respective modes. 
This simulation can be done by Alice and Bob as before because the energy test commutes with the measurements used for the key distillation.

The obvious open question is whether an active symmetrization is really needed in the case of a reverse reconciliation. We believe that this is not the case, and that this is a artifact of the current version of the Postselection technique. It is very natural to conjecture that a better version of the Postselection technique, exploiting all the symmetries of the protocol in phase-space could be sufficient to prove the security of the protocol without any need for active symmetrization. 

Finally, it should be noted that a similar kind of active symmetrization seems to be also required for protocols such as BB84 if Alice and Bob discard all the events where Bob's detectors did not register a photon. Indeed, postselecting on the events where the detectors clicked is a form of reverse reconciliation which breaks \textit{a priori} the symmetry of the protocol, i.e. its invariance under joint permutations of the subsystems of Alice and Bob in this case. For BB84, performing an active symmetrization of the data means drawing uniformly at random a permutation of size $n$, a task with complexity $O(n \log n)$. Of course, this is doable in theory, but not really practical. 

For this reason, the question of symmetrization is not relevant only for CV QKD, but also for protocols such as BB84 when considered in a practical setting where losses are large.

\end{widetext}

\end{document}